\documentclass[pra,10pt,twocolumn,superscriptaddress,floatfix,showpacs, longbibliography]{revtex4-1}

\usepackage[british]{babel}  
\usepackage[scaled=0.86]{berasans}  
\usepackage[colorlinks=true, allcolors=blue, urlcolor=blue]{hyperref}  
\usepackage{graphicx} 
\usepackage[babel]{microtype}  
\usepackage{amsmath,amssymb,amsthm,bm,amsfonts,mathrsfs,bbm} 

\usepackage{xspace}  
\usepackage{pgfplots}
\usepackage{xcolor,colortbl}
\usepackage{array}
\usepackage{dsfont}
\usepackage{bigstrut}

\newcommand{\ket}[1]{| #1 \rangle}

\newcommand{\bea}{\begin{eqnarray}}
\newcommand{\eea}{\end{eqnarray}}

\newcommand{\nrank}[1]{\mathrm{rank}_+(#1)} 
\newcommand{\psdrank}[1]{\mathrm{rank}_{psd}({#1})} 

\newtheorem{proposition}{Proposition}

\theoremstyle{definition}

\newtheorem{definition}{Definition}

\begin{document}
\title{Random Exclusion Codes: Quantum Advantages of Single-Shot Communication}

\author{Joonwoo Bae}
\email{joonwoo.bae@kaist.ac.kr}
\affiliation{School of Electrical Engineering, Korea Advanced Institute of Science and Technology (KAIST), $291$ Daehak-ro, Yuseong-gu, Daejeon $34141$, Republic of Korea}

\author{Kieran Flatt}
\email{kflatt@kaist.ac.kr}
\affiliation{School of Electrical Engineering, Korea Advanced Institute of Science and Technology (KAIST), $291$ Daehak-ro, Yuseong-gu, Daejeon $34141$, Republic of Korea}

\author{Teiko Heinosaari}
\email{ teiko.heinosaari@jyu.fi}
\affiliation{Faculty of Information Technology, University of Jyväskylä, Finland}

\author{Oskari Kerppo}
\email{oskari.kerppo@quanscient.com}
\affiliation{Quanscient Oy, Tampere, Finland}

\author{Karthik Mohan}
\email{karthikmohan@kaist.ac.kr}
\affiliation{School of Electrical Engineering, Korea Advanced Institute of Science and Technology (KAIST), $291$ Daehak-ro, Yuseong-gu, Daejeon $34141$, Republic of Korea}

\author{Andrés Muñoz-Moller}
\email{andres.d.munozmoller@jyu.fi}
\affiliation{Faculty of Information Technology, University of Jyväskylä, Finland}

\author{ Ashutosh Rai}
\email{ashutosh.rai@kaist.ac.kr }
\affiliation{School of Electrical Engineering, Korea Advanced Institute of Science and Technology (KAIST), $291$ Daehak-ro, Yuseong-gu, Daejeon $34141$, Republic of Korea}

\begin{abstract}

Useful applications of quantum information technologies can be found by identifying tasks in which quantum resources outperform their classical counterparts. In this work, we introduce a two-party communication primitive, {\it random exclusion code} (REC), which is a single-shot prepare-and-measure protocol where a sender encodes a random message into a shorter sequence and a receiver attempts to {\it exclude} a randomly chosen letter in the original message. We present quantum advantages in RECs in two ways: {\it probability} and {\it dimension}. We show that RECs with quantum resources achieve higher success probabilities than classical strategies. We verify that the quantum resources required to describe detection events of RECs have a smaller dimension than classical ones. We also show that a guessing counterpart, {\it random access codes} (RACs), may not have a dimension advantage over classical resources. Our results elucidate various possibilities of achieving quantum advantages in two-party communication. 

\end{abstract}

\maketitle


\section{Introduction}

One of the key challenges in quantum information theory is to address information tasks that demonstrate the distinct nature of quantum and classical resources. A gap in performance appearing from the distinction, even if it is small, may be amplified toward practical applications, namely, quantum advantages in communication \cite{winter}. It is hence of both fundamental and practical importance to identify a simplest information scenario, such as a point-to-point communication protocol, such that quantum advantages are present. The Holevo theorem, concerning the channel capacity with {\it i.i.d.} quantum states, was the first attempt along these lines but failed to show a quantum advantage \cite{Holevo2011, 651037, PhysRevA.56.131}. 

It is random access codes (RACs) that have verified the existence of quantum advantages in two-party communication \cite{Ambainis2009, Ambainis2024RAC, Tavakoli2015, carmeli2020quantum}. We emphasize that the setting is a single-shot scenario, in which a single outcome is obtained from a detection event, from which a receiver learns about a message. The guess is then correct with some probability. The quantum advantage is demonstrated by the fact that the guessing probability in RACs is higher if quantum, rather than classical, resources are used. 

Although the connection is not direct, RACs are closely related to quantum state discrimination \cite{Helstrom:1969aa, Bae_2015, Barnett:09, Bergou:2007aa}. Interestingly, quantum state exclusion, also known as anti-discrimination, which efficiently excludes a state from an ensemble, has been used as an effective tool to clarify the distinctions between quantum and classical resources in single-shot experiments \cite{Caves2002, PuseyPBR2012, PhysRevA.89.022336, HeKe18}. Quantum advantages can be obtained in terms of a success probability \cite{PhysRevLett.125.110402, PhysRevResearch.2.033374}, and a quantum-classical gap is also present in detection events. The Pusey-Barrett-Rudolph (PBR) theorem \cite{PuseyPBR2012}, for instance, rules out a certain interpretation of quantum states by pinpointing a detection event that cannot happen with quantum resources. 

Upon examining the fundamental tasks closely, it becomes clear that discrimination and exclusion are distinct in nature. Perfect discrimination of states may achieve perfect exclusion by reassigning events to different labels. The converse does not generally hold. Quantum state exclusion is thus operationally more general. State exclusion may therefore lead to greater possibilities of quantum advantages. 

As discussed above, efforts to distinguish quantum and classical resources in a single-shot scenario take one of several approaches. One is to compare the success probability for a single detection event. As with RACs, the probability with quantum resources may be higher than that with classical resources. The other approach is to take into account detection events. Along these lines, the PBR theorem signifies a detection event that is impossible with quantum resources but is possible with independent ontic preparations \cite{PuseyPBR2012, Harrigan:2010aa}.  Note that the Bell theorem is comparable to the latter in that it also rules out an ontic interpretation of quantum states; in contrast to the PBR result, two parties repeat a Bell experiment for estimating joint probabilities \cite{PhysicsPhysiqueFizika.1.195, Brunner2014}. It is worth mentioning that the quantum-classical gap from the Bell theorem has been developed as a practical tool for device-independent quantum information processing, such as randomness extraction \cite{Pironio:2010aa}, quantum key distribution \cite{PhysRevLett.98.230501} and state certification \cite{Supic2020selftestingof, Supic2021deviceindependent}. One would therefore expect useful applications to follow from other quantum-classical distinctions.

In this work, we present {\it random exclusion codes} (RECs) that leverage quantum state anti-discrimination to a communication primitive. Throughout, an $(n,m)$ REC of interest is a single-shot and prepare-and-measure protocol where a sender, Alice, encodes a random message of length $n$ using a set of $m$ letters into a $d$-dimensional system and transmits it. Bob, a receiver, realizes a decoding to exclude a chosen letter in the original message.

We show the quantum advantages of RECs in two ways: one is the probabilistic advantage for a single detection event, and the other is the dimension advantage for a collection of possible detection events. Firstly, a higher success probability in the exclusion task of a REC can be achieved with quantum resources than with classical ones. Secondly, quantum resources in a smaller Hilbert space achieve RECs that can be realized only in a larger classical dimension. Note that a classical dimension denotes the minimal number of distinguishable letters of an alphabet. 

The importance of the results is twofold. On the one hand, we relate state exclusion, a fundamental tool for verifying distinctions between the elements of quantum and classical theories \cite{Caves2002, Bandyopadyhyay2014Exclusion, Crickmore2020Elimination, PuseyPBR2012}, to a communication primitive that may find practical applications. On the other hand, we demonstrate the quantum-classical gap within a single-shot framework by highlighting the differences between quantum and classical supports for detection events, see also \cite{PuseyPBR2012}. We also find that RACs may not have a comparable dimension advantage. 

The article is structured as follows. We first remind the reader of some relevant results in the theory of quantum state exclusion and random access coding. We then introduce random exclusion codes and show that quantum resources outperform classical ones in the simplest non-trivial setting. We then analyse the required dimensional supports for perfect quantum and classical RECs and see that the former require fewer dimensions. Throughout, our results are contrasted with those from the equivalent RAC.

\section{Background}

\subsection{Quantum state exclusion}

In information processes, one is typically interested in determing that a datum takes a particular value from a known, discrete alphabet. Exclusion conversely refers to the opposite task of determining that a datum does not have a given value \cite{Crickmore2020Elimination}. The latter task is in general easier than the former, but recent works have shown that \emph{perfect} implementation of this task is a signature of quantum behaviour.
\par
In quantum state exclusion, a classical alphabet of size $ m$ will be associated with a set of quantum states $ \{ \rho_\alpha \}_{\alpha=0}^{ m-1}$ by a party Alice. She is given a value $\alpha$ at random and sends a system prepared in the appropriate state to a receiving party, Bob, who must announce one of the values which $\alpha$ is not. He determines his choice by implementing an $m$-outcome probability operator valued measurement (POVM) $\{ \pi_\beta \}_{\beta=0}^{ m-1}$.
\par
We can distinguish between different criteria for state exclusion. \emph{Perfect} state exclusion is that with zero-error, so that, whenever Alice prepares the state $\alpha$, Bob is guaranteed to announce a value $\beta \neq \alpha$. \emph{Uniform} state discrimination, on the other hand, refers to the case in which all allowed outcomes occur with uniform probability. We will return to this point in later sections of the manuscript. 

\begin{figure}
   \centering
\includegraphics[width=0.73 \linewidth]{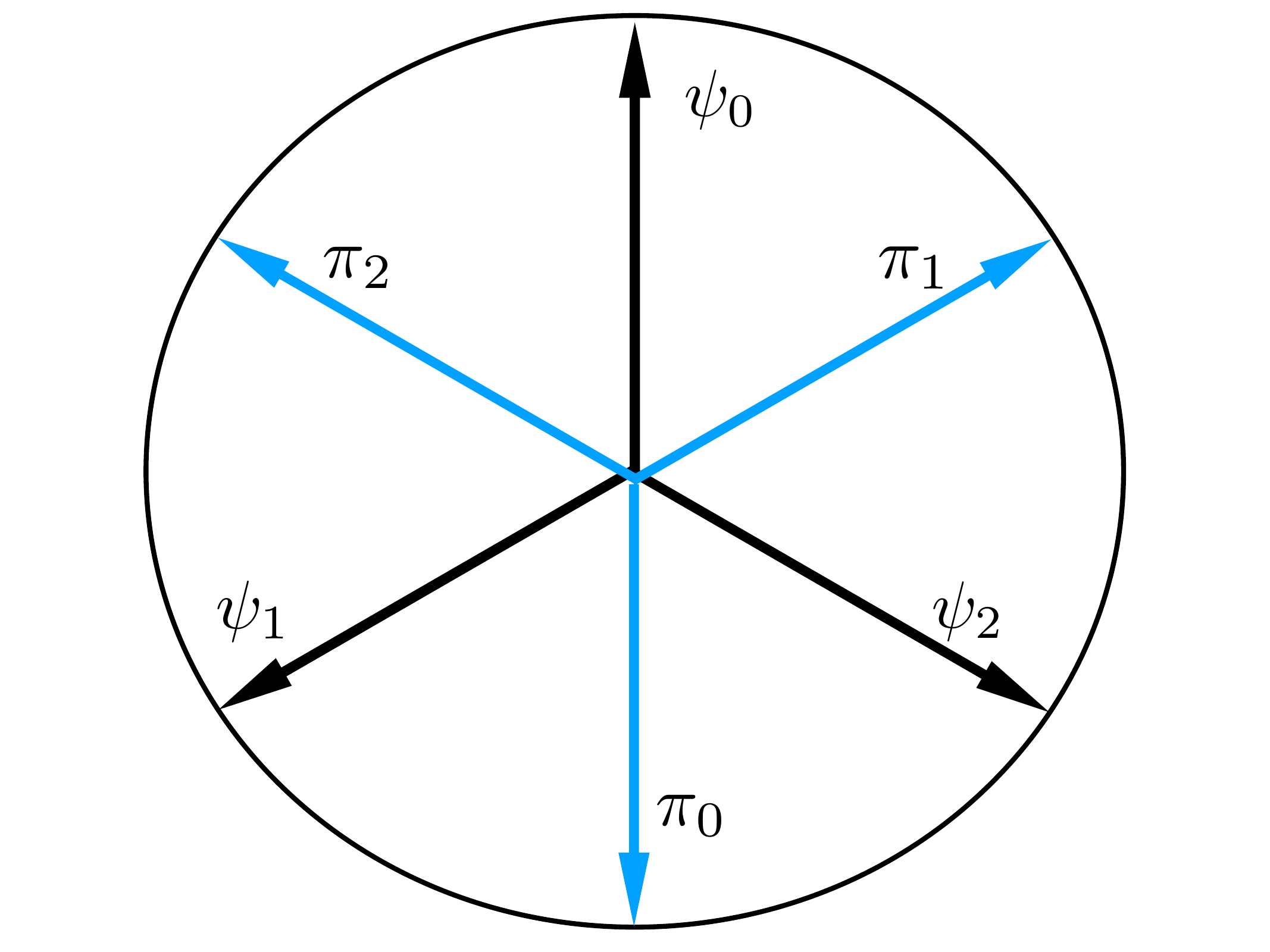}
    \caption{ Trine states (black) in Eq. (\ref{eq:trinem}) and an anti-trine measurement (blue) in Eq. (\ref{eq:antitrinem}) can realize an unbiased exclusion task. For instance, an outcome $\pi_0$ concludes state $\psi_1$ or state $\psi_2$ with equal probabilities $1/2$ while excluding $\psi_0$ unambiguously. }
\label{fig:trine}
 \end{figure}

The simplest non-trivial construction which is able to implement perfect and uniform state exclusion is the set of trine states
\begin{equation} \label{eq:trinem}
|\psi_\alpha \rangle  = \frac{1}{\sqrt{2}} \left( | 0 \rangle + e^{ i 2 \pi \alpha / 3} | 1 \rangle \right) \, \, \, \alpha = 0, 1, 2 \, ,
\end{equation}
and anti-trine measurements
\begin{equation} \label{eq:antitrinem}
\pi_{\beta} = \frac{2}{3} | {\psi}_{\beta}^{\perp}  \rangle \langle {\psi}_{\beta}^{\perp}  |,  \, \, \, {\beta}= 0, 1,2
\end{equation}
where $ \langle \psi_\beta | {\psi}_{\beta}^{\perp}  \rangle = 0 $. This set of states and measurements is shown in Fig. \ref{fig:trine} and has been used to implement key distribution \cite{Phoenix2000} and in showing gaps between classical and quantum behaviour \cite{Heinosaari2024Simple}. The defining feature in this example is symmetry. In similar symmetric constellations one finds the excluding POVM as a group covariant POVM \cite{HeKe18}.  

The problem of state exclusion was first studied by Caves et al. \cite{Caves2002} and came to wider attention when it was used to prove the PBR theorem regarding the $\psi$-onticity of quantum theory. Subsequent investigations have aimed to characterise the set of quantum ensembles which states can be unambiguously excluded \cite{Johnston2025tightbounds}, and related conclusive exclusion to various foundational problems in quantum theory \cite{Havlicek2020Exclusion, Srikumar2024contextuality, Stratton2024Choi} and the construction of zero-error communication channels \cite{Duan2016, Chiribella2025}. The key takeaway from these results has been that perfect state exclusion is a signature of quantum behaviour and, furthermore, may be exploited for improved communications.

\subsection{Random Access Codes}

Random access codes are a communication primitive widely studied in the context of quantum theory for both their information-theoretic applications and foundational significance~\cite{Ambainis2009, Ambainis2024RAC, Tavakoli2015, carmeli2020quantum}. Alice holds a string $a = a_1... a_{n}$ consisting of $n$ independent classical labels each taken equiprobably from an $m$-symbol alphabet  $\{0,\cdots,m-1\}$. This data is encoded into a system, either classical or quantum, of dimension $d$ which is then sent to the receiver, Bob. Bob's task is to read out some subset of the information. He is given a question label, $i$, and aims to respond with $b_i$ such that $b_i = a_i$. Measures of his ability to perform this task use either the worst-case probability of success or the average success probability, given by
\begin{equation} \label{eq:avscRAC}
{p}^{\mathrm{(RAC)}} =\frac{1}{n m^n} \sum_{a }\sum_{i=1}^{n} ~\mathrm{Prob} [ b_i =  a_i].
\end{equation}
If the communication system is quantum, Alice's encoding will map the data to a set of density matrices $\{ \rho_{a} \}$ and Bob's decoding will be a set of positive operator-valued measures (POVMs) $M_i$. If the system is instead classical, then the encoding and decoding are simply maps to and from dits. In this article, we are particularly interested in the case in which $n=2$ and the quantum and classical systems to be compared are both two-dimensional. In this restricted setting, upper bounds are known to be~\cite{Tavakoli2015,Ambainis2024RAC}:
\begin{equation}\label{eq:d-level-bound}
\begin{split}
   \max_{Q} p^{(RAC)} &= \frac{1}{2}(1 + \frac{1}{\sqrt{m}}) \\
 \max_{C}  p^{(RAC)} &= \frac{1}{2}(1 + \frac{1}{m}).
\end{split}
\end{equation}
where the former optimisation is over quantum strategies and the latter over classical ones. In this case, and in general, it holds that $\max_{Q} p^{(RAC)}  >  \max_{C}  p^{(RAC)}$.

\section{Random Exclusion Codes}

We begin with the notation for an $(n,m)$ REC. A message of Alice's is written by $a = a_1 ... a_{n}$  where each is chosen from a set of $m$ distinguishable letters, $a_i \in \{0,\cdots, m-1 \}$. An encoding prepares a $d$-dimensional system, and we let $b_j$ denote Bob's decoding of $a_j$ once he decides to read out $a_j$, see also Fig.~\ref{fig:scenario}. With an REC, Bob wants to find $b_j$ that is not equal to $a_j$. In this case, the probability of successfully excluding a correct value on average is given,
\bea \label{eq:avscREC}
{p}^{\mathrm{(REC)}} =\frac{1}{n m^n} \sum_{a }\sum_{i=1}^{n} ~\mathrm{Prob} [ b_i\neq a_i].
\eea
A quantum realization signifies a preparation of a $d$-dimensional quantum state and $m$-outcome measurements for Bob to obtain $b_j$. 

For RACs, in contrast to RECs, Bob wants to guess Alice's message and the average success probability is computed from Eq.~(\ref{eq:avscRAC}). In both the tasks, above in Eqs. (\ref{eq:avscRAC}) and (\ref{eq:avscREC}), parties aim to maximize the success probability over all strategies with given resources.

Note that state discrimination and state exclusion are equivalent for an alphabet of two letters: guessing one immediately excludes the other, and vice versa. For example, for binary messages $a_i \in \{0,1\}$, Bob can apply two-state discrimination and it holds that 
\bea
\max_S {p}^{ \mathrm{(RAC)} } =  \max_S {p}^{ \mathrm{(REC)}}, \label{eq:rel}
\eea
where the maximization runs over all strategies $S$. From this, one immediately concludes that RECs bear quantum advantages of RACs in the $(2,2)$ instance. However, the relation above no longer holds when an alphabet contains more than two letters as excluding a correct value does not yield a single outcome but rather multiple ones.

Let us then consider the next simplest case, the $(2,3)$ REC, where two-trit messages $ a=a_1a_2$ for $a_i\in \{0,1,2 \}$ are encoded to a binary system, for which a quantum (classical) realization is a single qubit (bit). In what follows, we demonstrate a quantum advantage by showing that Bob has a higher probability of excluding a message sent by Alice with quantum rather than classical resources. 
\par

\section{Quantum advantage in $(2,3)$-RECs}
We prove analytically that, if the communication system is a classical bit, the maximum value of the average probability of success in the considered task is $\frac{8}{9}$. We then show that if instead quantum systems of dimension two are used, a larger value of $ \frac{7+\sqrt{2}}{9}$ can be attained. It is worth mentioning here that for the considered REC, since the alphabet $\{0,1,2\}$ is non-binary, the corresponding RAC task is not equivalent, which can be seen from the fact that the maximum classical value for the RAC task is $\frac{5}{9}$ (see Result~3 in the Ref.~\cite{deba+PRA2023}). The inequivalence also holds for the quantum values, where the maximum success probability for the corresponding RAC is $\frac{4+\sqrt{2}}{9}$. \\
 
\subsection{Optimal bit implementation}
We first compute the maximum average success probability with classical resources, which can be summarised as follows: 
\begin{proposition}
The maximum classical average success probability in the $(2,3)$ random exclusion task, consisting of length~$n=2$ words using a size~$m=3$ alphabet and $d=2$ dimension communicated system, is $8/9$.
\end{proposition}

\begin{proof}
Let us begin by calculating the optimal value over all deterministic protocols. The two stages of a deterministic protocol are as follows: (i) The set of all possible words $\{a_1a_2: a_1,a_2\in\{0,1,2\}\}$ is partitioned into two parts $\{ \mathbb{P}_0, \mathbb{P}_1 \}$; if a word to Alice belongs to the part $\mathbb{P}_0$ ($\mathbb{P}_1$), she communicates to Bob $c=0$ ($c=1$). (ii) Bob, on receiving a question $i \in \{1,2 \}$, answers $b_i^c\in \{0,1,2\}$ following a deterministic strategy, i.e., a map $\{b_1^0,b_2^0,b_1^1,b_2^1\}\mapsto \{0,1,2\}$.

\par

Let us first note that, for a given deterministic strategy of Bob, two specific words follow, namely, the word $b_1^0 b_2^0$ and the word $b_1^1 b_2^1$. Now, let us consider that Alice receives either of these words. The message she communicates will then be determined by her partition, and there are four cases to consider:

\begin{enumerate}
    \item $b_1^0b_2^0, b_1^1b_2^1 \in \mathbb{P}_0$,
    \item $b_1^0b_2^0, b_1^1b_2^1 \in \mathbb{P}_1$,
    \item $b_1^0b_2^0 \in \mathbb{P}_0$ and $b_1^1b_2^1 \in \mathbb{P}_1$,
    \item $b_1^0b_2^0 \in \mathbb{P}_1$ and $ b_1^1b_2^1 \in \mathbb{P}_0$.
\end{enumerate}
We study each case in turn. In Case 1, when Alice receives the word $a = b_1^0 b_2^0$ she will communicate the value $c=0$ to Bob. For each question $i$ that he is asked he will output $b_i^0$, which was the value he aimed to exclude. Therefore, the protocol fails two times.  Note that, in this case, if instead she received the message $a=b_1^1 b_2^1$, Bob will end up answering $b_1^0$ and $b_2^0$ which need not be the same as the initial string; he can therefore succeed in these cases. Case 2 follows the same argument as Case 1. For Case 3, when Alice receives $a=b_1^0 b_2^0$, her message is $c=0$ and Bob's output will therefore be $b_i^0$, so he fails for both questions. Similarly, he will be incorrect whenever $b_1^1 b_2^1$ is the message received by Alice. Thus, in Case 3, Bob fails in at least four instances of the protocol.
\par
In Case 4, as in Case 3, it can never follow that $b_1^0 b_2^0 = b_1^1 b_2^1$ as this would mean that Alice puts the same word into two sides of the partition. Let us first suppose that $b_1^0 b_2^1 \neq b_2^0 b_1^1 $. Then, when Alice receives $a=b_1^0 b_2^1$ she will communicate either $c=0$ or $c=1$ depending on which partition it belongs to. In the former case, Bob's answer is in the string when his question is $i=1$ and in the latter case it is in the string when his question is $i=0$. Similarly, he will be incorrect in at least one position $i$ when the initial string is $a=b_2^0 b_1^1$. Case 4 therefore fails at least four times if $b_1^0 b_2^0 = b_1^1 b_2^1$. The final situation is that $b_1^0 b_2^1 = b_2^0 b_1^1$. This implies further that $b_1^0 = b_2^0$ and $b_1^1 = b_2^1$. On considering the words $a=b_1^0 b_1^1$ and $a=b_1^1 b_1^0$ , irrespective of which partition they belong to the protocol will fail at least once for each. Therefore, in Case 4 as well it is observed that the protocol fails to exclude in at least two instances.

\par
We can hence see that, in all four cases, for any deterministic strategy Bob answers incorrectly in at least two out of eighteen instances of the protocol. The maximal average success probability with deterministic strategies is therefore upper bounded by ${p}^{\mathrm{(REC)}}\leq 16/18 = 8/9$. This bound is, in fact, tight and can be achieved if Alice uses the partition $\mathbb{P}_0=\{11,12,21,22\}$ and $\mathbb{P}_1=\{00,01,02,10,20\}$, and if Bob's decoding is $b_i^c = c$. One can easily verify that this strategy gives the desired value. 
\par
Finally, we show that probabilistic strategies cannot perform better. Let's assume that there are $K$ deterministic strategies $S_j$ where $j\in\{1,\cdots,K\}$ and the parties have access to shared random variables $\{\lambda_j\}_{j=1}^{K}$ with probability distribution $p(\lambda_j)=p_j$. Note that any probabilistic strategy can be implemented with access to these random variables. If Alice and Bob apply strategy $S_j$ with probability $p_j$ we get a probabilistic strategy $S=\{p_j,S_j\}_{j=1}^K$. The average success probability with any such strategy $S$ is ${p}^{\mathrm{(REC)}}_{S}= \sum_j p_j {p}^{\mathrm{(REC)}}_{S_j} \leq  \sum_j p_j (8/9)= 8/9$. This completes the proof.
\end{proof}

\begin{figure}
    \centering
\includegraphics[width=0.98 \linewidth]{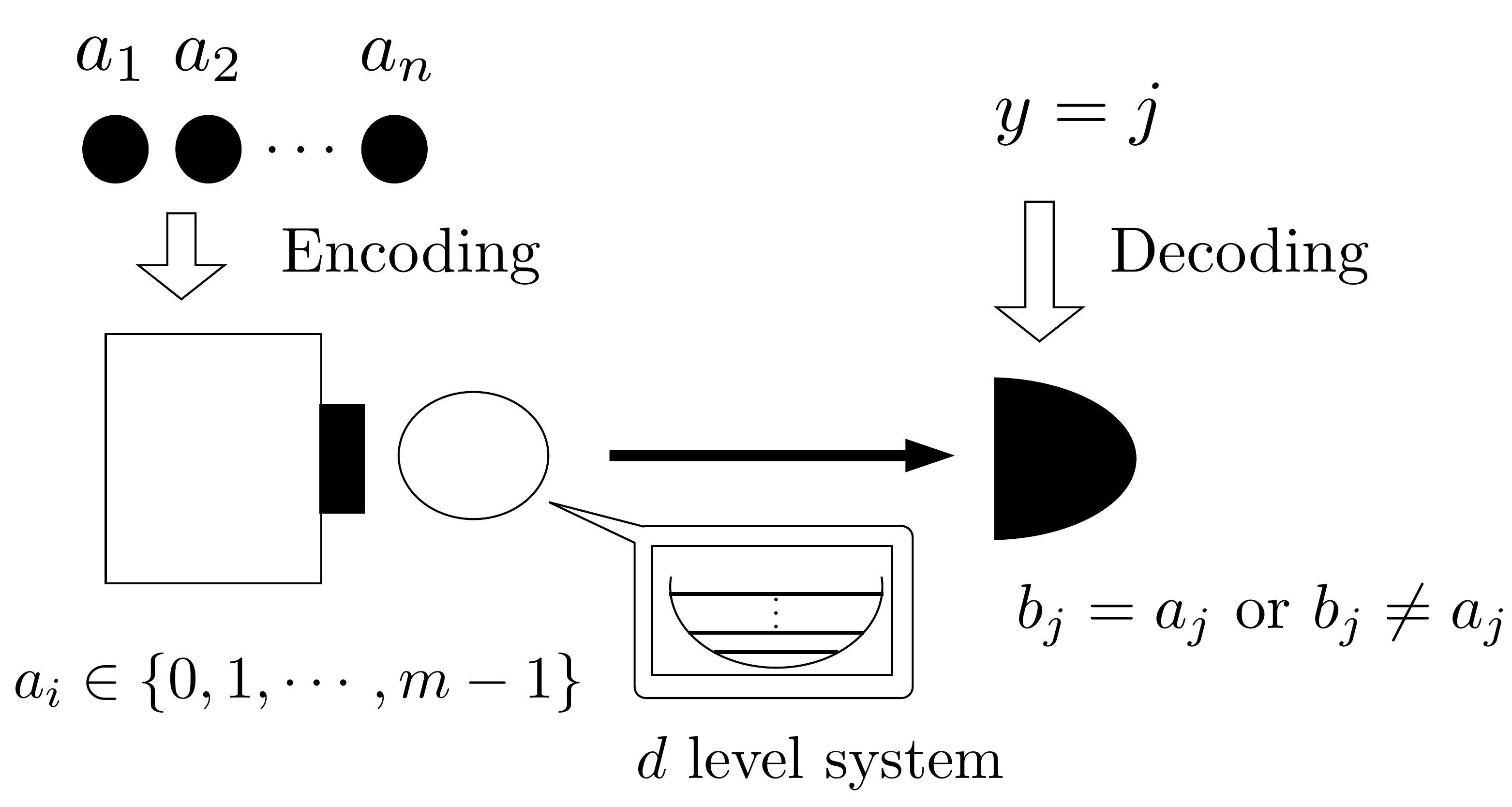}
    \caption{ In an $(n,m)$ REC, a sender prepares a message, $a=a_1\cdots a_n$,  of length $n$ where each is chosen from an alphabet of $m$ letters $a_i\in \{0,1,\cdots,m -1\}$. An encoding corresponds to a preparation of a single system of cadinality $d$, which is sent to Bob, who first decides which letter he will read, say the $j$-th one, and reads $b_j$. Bob wants to exclude the message sent by Alice, i.e., $b_j \neq a_j$. In RACs, parties want to have $b_j =a_j$. For these purposes, parties optimize strategies with quantum or classical resources.}
\label{fig:scenario}
\end{figure}

\subsection{ Qubit implementation and quantum advantage}

We next present a $(2,3)$ REC with single qubit states in the following,
\bea
\rho_a: = \rho_{a_1a_2} = |\psi_{a_1a_2}\rangle \langle \psi_{a_1a_2} |~\mathrm{ for}~  a_i \in \{0,1,2 \} \label{eq:s}
\eea
where
\bea
\ket{\psi_{00}} &=& \frac{  \ket{0} + e^{i \pi /4} \ket{1}}{\sqrt{2}},  ~~  \ket{\psi_{01}} = \frac{ \ket{0} + e^{-i \pi /4} \ket{1}}{\sqrt{2}},  \nonumber \\
\ket{\psi_{10}} &=& \frac{ \ket{0} + e^{i 3 \pi /4} \ket{1} }{\sqrt{2}},  ~  ~ \ket{\psi_{11}} = \frac{ \ket{0} + e^{-i 3 \pi /4} \ket{1} }{\sqrt{2}},  \nonumber \\
\ket{\psi_{20}} &=& \frac{\ket{0} + e^{i \pi /2} \ket{1}}{\sqrt{2}}, ~~ \ket{\psi_{21}} = \frac{\ket{0} + e^{ -i \pi /2} \ket{1} }{\sqrt{2}} \nonumber\\\ket{\psi_{02}} &=& \frac{ \ket{0} +  \ket{1}}{\sqrt{2}},   ~\ket{\psi_{12}} = \frac{ \ket{0} - \ket{1} }{\sqrt{2}},~\mathrm{and}~  \rho_{22} =   \mathbb{I}/2. \nonumber  
\eea
We construct a decoding with two-outcome measurements, 
\bea
M_1= \{ M_{10}, M_{11}, M_{12} \}~\mathrm{and}~M_2= \{ M_{20}, M_{21}, M_{22} \}~~~~ \label{eq:m}
\eea
where 
\bea
&&M_{10}= |+\rangle \langle +|,~M_{11}= |-\rangle \langle -|,~ M_{12}=0 \nonumber \\
&& M_{20}= |+i\rangle \langle +i|,~ M_{21}= |-i \rangle \langle -i|,~ \mathrm{and}~M_{22}=0, \nonumber
\eea
and $|\pm\rangle = (|0\rangle \pm |1\rangle)/\sqrt{2}$ and $|\pm i\rangle = (|0\rangle \pm i |1\rangle)/\sqrt{2}$. Notice that the states and measurements above are on a great circle of the Bloch sphere. We display the states and measurements of our scheme in in Fig. \ref{fig:23recsm}.

\begin{figure}[t!]
		\begin{center}
			\includegraphics[scale=0.18]{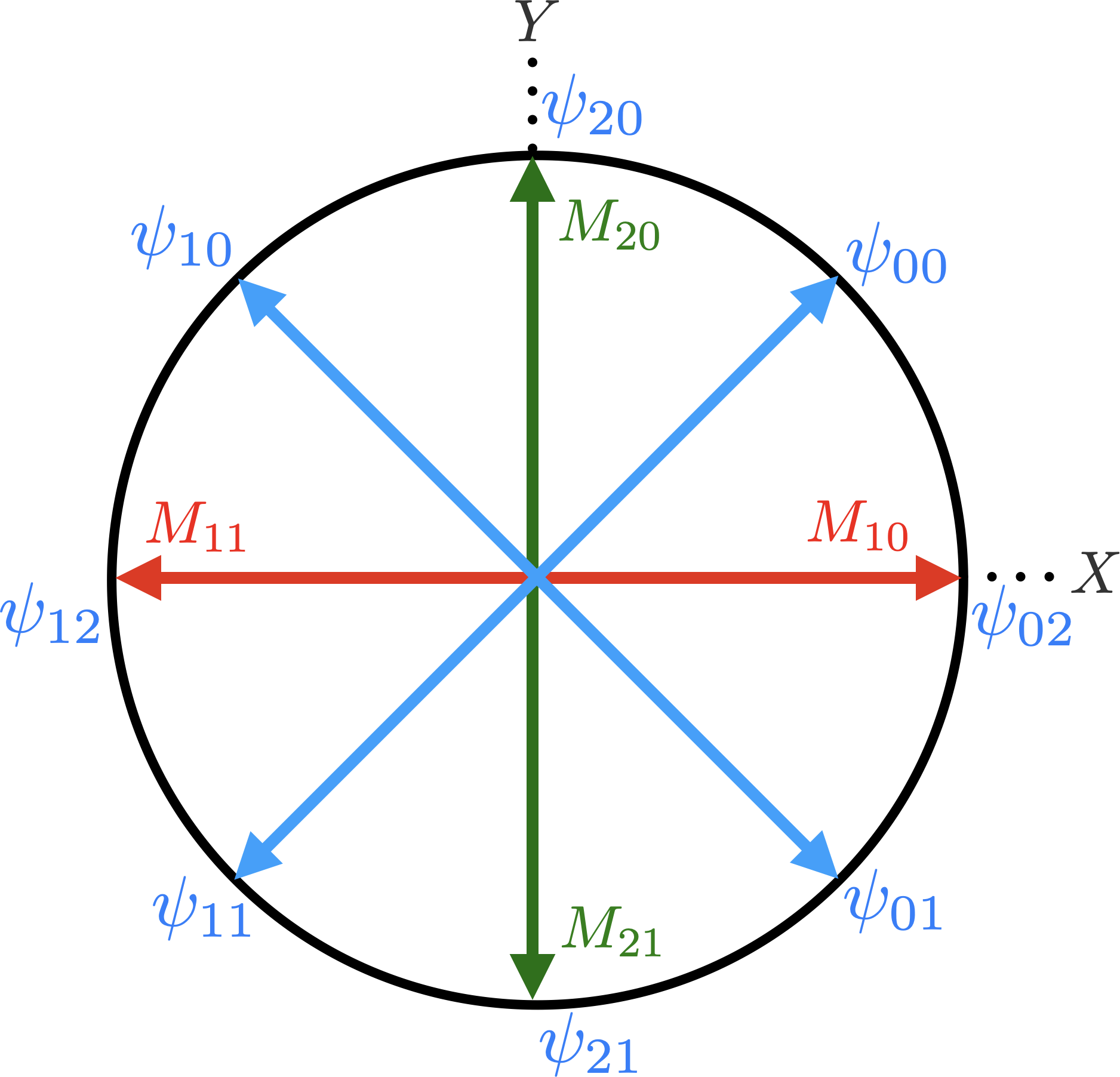}	
   \caption{Illustration of states and measurements for the scheme discussed in text. States are displayed in light blue and the measurements $\mathcal{M}_1$ and $\mathcal{M}_2$ in red and green respectively.} \label{fig:23recsm}
		\end{center}
	\end{figure}

The quantum protocol works as follows. Alice sends a state $\rho_{a_1a_2}$ and Bob first decides which one, $a_1$ or $a_2$, he is going to read. For the first trit $a_1$, he applies measurement $M_1$, and for the second, $a_2$, he applies $M_2$. The probability of successfully excluding a one-bit message of Alice is given by,
\begin{eqnarray} 
   p^{ \mathrm{(REC)}}_{Q} &=&\frac{1}{18} \sum_{a}\sum_{i=1}^2{\rm Tr} [ (M_{i(a_i\oplus 1)}+M_{i(a_i\oplus 2)})~\rho_{a_1 a_2}  ] ~~\label{eq:avsq}
\end{eqnarray}
where $\oplus$ represents addition modulo three. It follows that  
\bea
p^{ \mathrm{(REC)}}_{Q} = \frac{7+\sqrt{2}}{9}  \label{eq:qbound}
\eea
which is larger than the maximum classical value $8/9$ (Proposition~1.). Thus, the utility of quantum resources in RECs is demonstrated.

\subsection{ Maximum quantum advantage}
There may exist states and measurements other than those in Eqs. (\ref{eq:s}) and (\ref{eq:m}), in particular using POVMs containing no null elements, such that the success probability is even higher. That is, one can ask if the bound in Eq. (\ref{eq:qbound}) is maximal. Our answer is in the affirmative. Firstly, we find that to maximize the success probability with qubit transmission, it suffices to restrict states and measurements to a half-plane on a Bloch sphere. This also holds true even if parties aim at a guessing task, and hence, applies to RACs with qubits. Secondly, extensive numerical optimization of the average success probability finds Eq. (\ref{eq:qbound}) with two POVM elements are zeros. We detail the calculations in Appendix~\ref{appendA}. The gap between quantum and classical bounds is
\bea
\Delta_{QC}^{(\mathrm{REC})} = \max_Q p^{(\mathrm{REC})}  -\max_C p^{(\mathrm{REC})}  = \frac{\sqrt{2}-1}{9} , \label{eq:del}
\eea
which is approximately $0.046$.

For comparison, let us also consider a $(2,3)$ RAC with single-qubit transmission. An upper bound to success probabilities with classical strategies is known as $5/9$ \cite{deba+PRA2023}. We compute the maximum average probability with quantum resources, applying the aforementioned result of reducing to a great circle. We find that a relabelling of our quantum protocol for the REC gives the quantum maximum for the RAC (see Appendix \ref{appendA}); messages $a=a_1a_2$ are encoded in the states orthogonal to ones given in Eq. (\ref{eq:s}) and the same measurements as in Eq. (\ref{eq:m}) are applied for decoding. The quantum bound is obtained as $(4+\sqrt{2})/9$. Interestingly, the gap is identical to Eq. (\ref{eq:del}), i.e.,
 \bea
\Delta_{QC}^{(\mathrm{RAC})} =  \Delta_{QC}^{(\mathrm{REC})}\label{eq:equal}
\eea
for the considered $(2,3)$ RAC and $(2,3)$ REC.
\par
To summarise, we have characterised bit and qubit implementations of the $(2,3)$ REC and shown that using quantum systems allows for a higher average probability of success. We also contrasted the behaviour of the $(2,3)$ RAC. 

\section{Dimensional Advantages}

In what follows, we are motivated by the observation that a single-shot scenario gives one outcome among possible detection events. This raises the question of which dimensional resources are required to observe distinctions between classical and quantum behaviours in the single-shot setting. Having characterised above the optimal performance of the $(2,3)$ REC in two-dimensional systems, we now turn our attention to a more general setting. We investigate the minimal quantum and classical resources to describe possible detection events in an errorless $(2,3)$ REC. We first note that there is no difference in behaviour for the standard scenario. This further motivates us to add an extra constraint, uniformity of outcomes, to our protocol. Classical implementations of this modified protocol require nine-dimensional systems, whereas quantum systems of four-dimensional behaviour, showing a dimensional advantage in uniform RECs. This behaviour is contrasted with RACs, where no such advantage can be observed.
\par
\subsection{Perfect random exclusion codes}
We first note that no dimensional advantage can be observed for zero-error implementations of the basic $(2,3)$ REC which has been discussed so far. To do this, we note that there exists a three-dimensional classical, and therefore also quantum, strategy which may implement this task with no error. The explicit strategy is for Alice to encode her messages using the partition
\begin{equation}
    \begin{split}
    \mathbb{P}_0 &= \{ 11, 12, 22 \} \\
    \mathbb{P}_1 &= \{ 00, 02, 20, 22 \} \\
    \mathbb{P}_2 &= \{ 01, 10 \}
    \end{split}
\end{equation}
and send the relevant trit value to Bob. If he announces the value he has received, he will exclude without error. As we saw in the previous section, there exists neither a classical nor a quantum strategy that is able to achieve an average success probability of one using a two dimensional implementation. We can thus conclude that three is the minimal dimension requirement in both cases.

\subsection{Uniform random exclusion codes}
In order to witness a dimensional advantage in the $(2,3)$ REC, we must therefore analyse a modified form of the problem. This extra requirement will be \emph{uniformity} of the outcomes, meaning that the correct answers by Bob all occur with equal probability. This constraint is regularly considered in the literature on exclusion tasks (see, e.g., Ref \cite{Heinosaari2024Simple, PhysRevResearch.2.033374}) in order to avoid certain trivial solutions to the problem of state exclusion. It is in fact uniform state exclusion that can be considered a strong signature of quantum, rather than classical, behaviour. The constraint can be furthermore justified if we consider that Alice and Bob are trying to convince a referee, who provides Alice with her string, of their ability to perform quantum exclusion. If their results are non-uniform, then the referee will be able to use their strategy against them when the task is performed over several rounds. For example, in our optimal qubit scheme, Bob will never announce two as his answer. The referee can thus always send strings that do not contain this value, lowering their probability of success. To summarise, our aim here will be to determine the minimum quantum and classical dimensions required to implement a zero-error, uniform $(2,3)$ REC.
\par
The natural framework for studying dimensional requirements is that of communication matrices \cite{Kerppo2023Communication}, which we review here for the reader's convenience. A communication matrix is a row-stochastic matrix encoding the conditional probabilities of prepare-and-measure experiments. If the inputs are a set of classical labels $ \alpha \in [ 0, m-1]$ and the outputs a set of classical labels $ \beta \in [0, n-1]$ then the communication matrix is an $m \times n$-sized row-stochastic matrix $C$ in which 
\begin{equation}
C_{\alpha \beta} = {\rm Prob}_{M|P}(\beta|\alpha),
\end{equation}
where the subscripts $P$ and $M$ denote preparation and measurement respectively. For instance, the communication matrix of a binary state discrimination task in which two states are \emph{perfectly} distinguished would be the diagonal matrix $S_2 = {\rm diag}[1,1]$. Similarly, a uniform state exclusion task with three states excluded in an unbiased manner would correspond to the communication matrix \begin{align} \label{eq:a3matrix}
    A_3 = \frac{1}{2}\begin{bmatrix}
    0 & 1 & 1 \\ 1 & 0 & 1 \\ 1 & 1 & 0
\end{bmatrix}.
\end{align}
This can be implemented in a quantum system by the trine-antitrine construction given earlier in Eqs. \ref{eq:trinem} and \ref{eq:antitrinem} .
\par
The use of communication matrices in quantum information science regards the different sets $C^d_{cl}$ and $C^d_{q}$ which can be generated by $d$-dimensional classical and quantum systems respectively. We recall that the quantum dimension is that of the associated Hilbert space and the classical one is defined by the number of perfectly distinguishable letters in a communication channel. An important theorem due to Frenkel and Weiner \cite{Frenkel2015Classical} tells us that the closure of the quantum and classical two sets is equal whenever they are generated by systems with matching dimension. The question becomes more intriguing if we ask how many dimensions are required of a quantum or classical system to implement a given communication matrix. With this in mind, it has been shown that the minimal quantum and classical dimensions follow from two different properties of a given communication matrix. Firstly, it is known that the minimum classical dimension required is given by the nonnegative rank:
\par
\definition{\textbf{Nonnegative rank.} A matrix $C$ is said to be nonnegative if all elements are nonnegative. The nonnegative rank \cite{Cohen1993Nonnegative} of an $n\times m$ matrix $C$, denoted $\nrank{C}$, is defined as the smallest number $k$ such that there exists an $n \times k$ nonnegative matrix $L$ and $k \times m$ nonnegative matrix $R$ such that $LR=C$. }
\par
\vspace*{0.2cm}
Similarly, the minimum dimension of a quantum system is given by the positive semi-definite rank:
\par
\definition{ \textbf{Positive semidefinite rank.} The positive semidefinite rank \cite{Fawzi2015Positive} (also \emph{psd rank}) of a $n\times m$ nonnegative matrix $C$, denoted $\psdrank{C}$ is defined as the smallest integer $k$ such that there exist positive semidefinite $k \times k$ matrices $A_1, \ldots, A_n$ and $B_1, \ldots, B_m$ such that $C_{ij} = {\rm Tr}[A_i B_j]$.}
\par
\vspace*{0.2cm}
Let us consider these quantities for the three-input unbiased exclusion task seen in $A_3$ (Eq. \ref{eq:a3matrix}). The psd rank of this matrix is found to be $\psdrank{A_3}=2$, which concurs with our observation that it may be implemented using qubits. On the other hand, the non-negative rank is $\nrank{A_3}=3$, meaning that classical bits may not implement unbiased exclusion without error.
\par
With this in mind, let us find the minimal dimension systems required for our task. The communication matrix for the considered uniform $(2,3)$ REC with unit success probability can be obtained as follows,
\begin{align} \label{eq:d3matrix}
D_3 &= \frac{1}{4}
\begin{bmatrix}
    0 & 0 & 0 & 0 & 1 & 1 & 0 & 1 & 1 \\
    0 & 0 & 0 & 1 & 0 & 1 & 1 & 0 & 1 \\
    0 & 0 & 0 & 1 & 1 & 0 & 1 & 1 & 0 \\
    0 & 1 & 1 & 0 & 0 & 0 & 0 & 1 & 1 \\
    1 & 0 & 1 & 0 & 0 & 0 & 1 & 0 & 1 \\
    1 & 1 & 0 & 0 & 0 & 0 & 1 & 1 & 0 \\
    0 & 1 & 1 & 0 & 1 & 1 & 0 & 0 & 0 \\
    1 & 0 & 1 & 1 & 0 & 1 & 0 & 0 & 0 \\
    1 & 1 & 0 & 1 & 1 & 0 & 0 & 0 & 0 \\
\end{bmatrix},
\end{align}
where each element corresponds to the probability of Bob inferring a given string using his observed outcome. This matrix can be decomposed as $D_3 = A_3 \otimes A_3$.
\par
We first note that the nonnegative rank of $D_3$ is equal to nine. As the nonnegative rank is lower bounded by the rank, and the above matrix is full rank, it follows that $\nrank{D_3}=9 $. Thus, a classical channel with at least nine perfectly distinguishable symbols is required to realise $(2,3)$ RECs in a perfect, uniform manner.
\par 
A minimal quantum dimension can be found by upper and lower bounds to the positive-semidefinite rank. We consider a quantum realization to ensure an upper bound and then compute a general lower bound. The two bounds are found to coincide. 
\par
A quantum realization using two qubits is as follows. Alice encodes $a_1$ of her string into one qubit, and $a_2$ into another, as one of the three trine states in Eq. (\ref{eq:trinem}). That is, the encoding is
\begin{equation}
\rho_{a_1 a_2} = |\psi_{a_1} \rangle \langle \psi_{a_1} | \otimes |\psi_{a_2} \rangle \langle \psi_{a_2} |,
\end{equation}
On receiving the systems, Bob applies an anti-trine measurement (see Eq. (\ref{eq:antitrinem})) onto the relevant qubit. That is, if he wishes to perform exclusion on the first bit of the string he uses the POVM
\begin{equation}
M_1 = \{ \frac{2}{3} | \psi^{\perp}_{b_1} \rangle \langle \psi^{\perp}_{b_1}  | \otimes \mathbb{I}  \}_{b_1 = 0}^{2},
\end{equation}
and likewise, to perform exclusion on the second bit, the POVM 
\begin{equation}
M_2 = \{ \frac{2}{3} \mathbb{I} \otimes | \psi^{\perp}_{b_2} \rangle \langle \psi^{\perp}_{b_2}  |  \}_{b_2 = 0}^{2},
\end{equation}
He can hence exclude the state, and therefore classical label, unambiguously. Furthermore, due to the symmetry in the ensemble, the remaining outcomes occur equiprobably, satisfying our second constraint. A two-qubit space with a dimension of four, therefore, suffices for the purpose. We have 
\bea
\mathrm{rank}_{psd} (D_3) \leq 4. \label{eq:ub}
\eea 
In Ref. \cite{Lee2017Some}, a lower bound to psd ranks for positive doubly stochastic matrices is shown. Applying the bound to $D_3$, we  have
\begin{equation}
    \psdrank{D_3} \geq \max_q \frac{1}{\sum_{i,j} q_i q_j F(  \mathrm{col}_i (D_3 ), \mathrm{col}_j (D_3 ))^2} \nonumber
\end{equation} 
where $ \mathrm{col}_i (D_3 )$ is the $i$-th column of $D_3$, the maximum runs over probability vectors $q = [q_1,\cdots, q_9] $ and $F(a,b) = \sum_i \sqrt{a_i b_i}$. We compute that the lower bound is
\bea
\psdrank{D_3} \geq 4, \nonumber 
\eea
which occurs if $q_i= 1/9$ for $i=1,\cdots,9$. As upper and lower bounds coincide, we conclude a minimal quantum dimension $4$. The $(2,3)$ REC can be realized on a smaller quantum dimension than a classical one, showing a distinction in quantum and classical supports for detection events.
\par
\subsection{Dimension requirements for (2,3) RAC}
For contrast, we also consider the required minimal quantum and classical dimensions for a $(2,3)$ RAC.  Again, we will be concerned with a perfect implementation of a random access code, i.e., one with an average success probability of one. In order to implement RACs perfectly, one would expect that all possible bit strings can be encoded onto perfectly distinguishable quantum or classical states. This can be proven if one notices that the communication matrix for noiseless communication is given by 
\begin{equation}
\mathbb{I}_9 = \mathrm{diag}[1,1,1,1,1,1,1,1,1] 
\end{equation}
which is on the support of a dimension $9$. For this matrix, it can be shown that $\mathrm{rank}_{+} (  \mathbb{I}_9 )  = \mathrm{rank}_{psd} ( \mathbb{I}_9) =9$ following a known theorem that the nonnegative and positive semidefinite ranks coincide with the dimension for the identity matrix \cite{Fawzi2015Positive}. Hence, no dimensional advantage is achieved in the perfect $(2,3)$ RAC. In other words, if a protocol can implement random access perfectly, we can implement exclusion perfectly by relabelling outcomes. The converse, however, is not true: one cannot realize random access with certainty on the quantum support where random exclusion is perfect. In general, one would expect that perfect \emph{access} of data would require at least as many perfectly distinguishable states as the number of bit strings available to participants, whereas, as we have shown, this does not hold for exclusion. This result emphasises that exclusion is a strong signature of quantum behaviour.

\section{Discussion and Conclusions}
Finally, RECs can be generalized to higher dimensions. We analyse in Appendix~\ref{appendB} the $(2,m)$ REC with a single qubit and show quantum advantages for all values of $m\geq 2$. Namely, we have the maximum of the success probability with classical resources, 
\bea
\max_C p^{(REC)} = 1-\frac{1}{m^2}. 
\eea
We extend the quantum strategy considered for $m=3$ to a larger $m$ and have,
\bea
p_{Q}^{(REC)} = 1-\frac{2-\sqrt{2}}{m^2}, 
\eea
which is strictly higher than the classical one. Quantum advantages remain for all $m\geq 2$. It should be emphasised that the upper bound we have provided for quantum (2, m) RECs may not be optimal.

In conclusion, we have presented a framework for identifying quantum advantages of a single-shot protocol by introducing a communication primitive, RECs, which relies on an exclusion task. A single-shot scenario has a detection event where outcomes occur probabilistically. Quantum advantages are provided by a probability and a dimension. Firstly, the probability that the exclusion task is successful from a single detection event is higher with quantum strategies than with classical ones. Secondly, by exploiting communication matrices, we have shown that the quantum resources of a smaller dimension can realize RECs with a larger alphabet. The dimensional quantum advantage may be lacking in RACs; a protocol for a RAC implies an implementation of a REC by relabelling, but the converse does not hold, similarly to the distinction between state discrimination and state exclusion.  

Regarding future directions, it would be interesting to investigate general properties of $(n,m)$ RECs and their relations to quantum state exclusion. In particular, we leave it open to characterize optimal measurements for RECs and quantum state exclusion to find general relations between them. We also conjecture that RECs may have quantum advantages for all $n,m\geq2$. In addition, it is also interesting to combine the probability and the dimension advantages of quantum resources to amplify the utility. We envisage that our results would lead to practical quantum information applications that outperform their classical counterparts. 

\section*{ACKNOWLEDGEMENTS}

We thank Rutvij Bhavsar for discussions regarding this project. This work is supported by Business Finland funded project BEQAH, the National Research Foundation of Korea (Grant No. NRF-2021R1A2C2006309, NRF-2022M1A3C2069728, RS-2024-00408613, RS-2023-00257994) and the Institute for Information \& Communication Technology Promotion (IITP) (RS-2023-00229524, RS-2025-02304540). 

\appendix

\section{Maximum quantum success probability in $(2,3)$ REC}\label{appendA}
 Here we provide a proof of the maximum quantum success probability in a $(2,3)$ REC task. To this
end, in the following, first we show by applying a known result for an optimal quantum random access code that if
only two outcome measurements are applied for decoding, the quantum protocol we have provided in the main text is an
optimal one.

\begin{proposition}\label{prop:optimal-qubit-two-outcome}
    Over the set of all two-outcome POVMS, the maximum average success probability in the quantum random exclusion task is $\frac{7+\sqrt{2}}{9}$.
\end{proposition}   

\begin{proof}
Without loss of generality, we consider the two 2-outcome POVMs applied in the protocol are: $M_1\equiv M_{10}+M_{11}=\mathbb{I}$ ~(we take $M_{12}=0$) and $M_2\equiv M_{20}+M_{21}=\mathbb{I}$ ~(we take $M_{22}=0$). Let us partition the set of all words $a=a_1a_2\in \mathcal{A}^2$ into two parts: $\mathbb{P}=\{00,01,10,11\}$ and $\mathbb{P}^c=\{02,20,12,21,22\}$. Then, from the Eq.~(\ref{eq:avsq}) we get
\begin{align}
  {p}^{\mathrm{(REC)}}_{Q} &\!=\! \frac{1}{18} \sum_{a}\sum_{i=1}^2{\rm Tr} [ (\mathbb{I}-M_{ia_i}) \rho_{a_1 a_2}  ] \nonumber\\ 
    &= \frac{4}{9}\left\{\frac{1}{8}\sum \limits_{a\in \mathbb{P}} \sum\limits_{i=1}^2 {\rm Tr} \left[(\mathbb{I}-M_{ia_i}) \rho_{a_1 a_2}\right]\right\}\nonumber \\
    &+\frac{1}{18}\sum \limits_{a\in \mathbb{P}^c}\left\{\sum\limits_{i=1}^2 {\rm Tr} [ (\mathbb{I}-M_{ia_i}) \rho_{a_1 a_2}  ] \right\} \nonumber\\
    &\leq \frac{4}{9}\left\{\frac{1}{8}\sum \limits_{a\in \mathbb{P}} \sum\limits_{i=1}^2 {\rm Tr} \left[(\mathbb{I}-M_{ia_i}) \rho_{a_1 a_2}\right]\right\} +\frac{5}{9} \nonumber\\
     &\leq \frac{4}{9}\left\{\max\limits_{\{\substack{M_{ia_i},\\\rho_{a_1a_2}}\}}~\frac{1}{8}\sum \limits_{a\in \mathbb{P}} \sum\limits_{i=1}^2 {\rm Tr} \left[(\mathbb{I}-M_{ia_i}) \rho_{a_1 a_2}\right]\right\} \!+\!\frac{5}{9} \nonumber\\
      &= \frac{4}{9}\left\{\max\limits_{\{\substack{\tilde{M}_{ia_i},\\\rho_{a_1a_2}}\}}~\frac{1}{8}\sum \limits_{a\in \mathbb{P}} \sum\limits_{i=1}^2 {\rm Tr} \left[\tilde{M}_{ia_i}~ \rho_{a_1 a_2}\right]\right\} +\frac{5}{9} \nonumber \\
      &~~~~~~~(\mbox{we~rewrite}~\mathbb{I}-M_{ia_i}=\tilde{M}_{ia_i}~),\nonumber\\
      &\leq \frac{4}{9}\left\{\frac{1}{2}\left(1+\frac{1}{\sqrt{2}}\right)\right\}+\frac{5}{9}=\frac{7+\sqrt{2}}{9}.
\end{align}
In the last step above, we simply apply the known result, i.e., in a $(2,2)$ RAC, the quantum maximum for the average success probability is $(1+1/\sqrt{2})/2$. Since our quantum protocol reaches the derived quantum upper bound, we find it to be optimal over the set of all two-outcome measurement protocols.

\end{proof}

In general, the optimization problem consists of maximizing the figure of merit over all possible three-outcome POVMs on a qubit space. To find optimal states and measurements for the purpose, let us rewrite the average success probability as follows, 
\begin{eqnarray} 
    {p}^{\mathrm{(REC)}}_{Q} &=& \frac{1}{36} \sum_{a} {\rm Tr} [\sum_{i=1}^2 [ (\mathbb{I}-M_{ia_{i}}) \otimes \rho_{a_1 a_2}^{T})  |\phi^+\rangle\langle \phi^+ | ],\nonumber 
\end{eqnarray}
where $|\phi^+\rangle = (|00\rangle + |11\rangle)/\sqrt{2}$. Note that a state $|\phi^+\rangle$ is permutationally invariant, which means that the states and the measurements that maximise the success probability should be in a symmetric subspace, supported by a quantum $2$-design. This implies that optimal states and measurements fulfil the condition,
\begin{eqnarray} 
\sum_{i=1}^2 (\mathbb{I}-M_{ia_{i}}) \propto \rho_{a_1 a_2}^{*}. \label{eq:para}
\end{eqnarray}
Applying the above, we will then show that to achieve the maximal success probability, it is sufficient to optimize a function which depends only on the measurement parameters (i.e., for any given decoding measurements, one can always choose an optimal set of states derived from the measurement). All these simplify the construction of states and measurements. \\

First, recall that we need to maximize  ${p}^{\mathrm{(REC)}}_{Q}$ in Eq.~(\ref{eq:avsq}) in the main text. From the objective function, one can see that the maximum is achieved with extremal (pure) states and measurements. Then let us consider a set of nine pure qubit states
\begin{eqnarray} 
    \rho_{a_1a_2}&=&\frac{1}{2} (\mathbb{I} +\hat{n}_{a_1a_2}\cdot \vec{\sigma}), 
\end{eqnarray}
where $\hat{n}_{a_1a_2}$ for $a_1,a_2\in \{0,1,2\}$ are unit vectors on the Bloch sphere, and two measurements corresponding to extremal POVMs which are given by rank one projectors, i.e.,
\begin{eqnarray} 
    M_{1a_1}&=&  r_{1a_1}\left(\frac{\mathbb{I} +\hat{m}_{1a_1}\cdot \vec{\sigma}}{2}\right),\\[5pt]
     M_{2a_2}&=&r_{2a_2}\left(\frac{\mathbb{I} +\hat{m}_{2a_2}\cdot \vec{\sigma}}{2}\right),\label{eq:extr}
\end{eqnarray}
 where $\hat{m}_{1a_1}$ ($\hat{m}_{2a_2}$) are unit vectors on the Bloch sphere in a plane passing through the origin such that weights $0\leq r_{1a_1}\leq 1$ ($0\leq r_{2a_2}\leq 1$) and $\sum_{a_1} r_{1a_1}=2$~ ($\sum_{a_2} r_{2a_2}=2$). On rearranging the terms in ${p}^{\mathrm{(REC)}}_{Q}$ in Eq.~(\ref{eq:avsq}), the expression for quantum success probability becomes
\begin{eqnarray}
{p}^{\mathrm{(REC)}}_{Q}&=&1-\frac{1}{18}\sum_{a}{\rm Tr}\left[ (M_{1a_1}+M_{2a_2})\rho_{a_1a_2}\right]\nonumber\\
&=&\frac{2}{3}-\frac{1}{36}\sum_{a}\left\{\hat{n}_{a_1a_2}\cdot(r_{1a_1}\hat{m}_{1a_1}+r_{2a_2}\hat{m}_{2a_2})\right\},\nonumber\\
&~&\mbox{then~application~of~Eq.~(\ref{eq:para})~gives} \nonumber \\
{p}^{\mathrm{(REC)}}_{Q}&\leq& \frac{2}{3}+\frac{1}{36}~\sum_{a}\vert r_{1a_1}\!\hat{m}_{1a_1}\!+\!r_{2a_2}\hat{m}_{2a_2}\vert.
\end{eqnarray}
 Note that, very similar to the the considered REC task, on analysing the corresponding RAC, through Eq.~(\ref{eq:avscRAC}), for maximum quantum success, we get an achievable upper bound function as follows:
    \begin{eqnarray}
{p}^{\mathrm{(RAC)}}_{Q}
&\leq& \frac{1}{3}+\frac{1}{36}~\sum_{a}\vert r_{1a_1}\!\hat{m}_{1a_1}\!+\!r_{2a_2}\hat{m}_{2a_2}\vert.
\end{eqnarray}
Thus, the optimization problems both for the quantum REC and RAC reduce to maximizing the same objective function
\begin{eqnarray} 
    f &=& \sum_{a}|r_{1a_1}\hat{m}_{1a_1}\!+\!r_{2a_2}\hat{m}_{2a_2}|.
\end{eqnarray}
From the steps above one can easily see that for arbitrary fixed measurements, one can always choose states such that the upper bound is attained. 
Therefore, it is sufficient to maximize the quantity
\begin{eqnarray} 
    f \!&=&\! \sum_{a_2\!=\!0}^{2}\!\sum_{a_1\!=\!0}^{2}|r_{1a_1}\hat{m}_{1a_1}\!+\!r_{2a_2}\hat{m}_{2a_2}| \nonumber \\[5pt]
     \!&=&\! \sum_{a_2\!=\!0}^{2}\!\sum_{a_1\!=\!0}^{2}\!\!\left\{r_{1a_1}^2\!\!+\!r_{2a_2}^2\!\!+\!2 r_{1a_1}\!r_{2a_2}\!\cos(\phi_{1a_1}\!\!-\!\phi_{2a_2})\sin\theta_{2a_2}\right\}^{\frac{1}{2}}, \nonumber \\ \label{eq:f}
\end{eqnarray}
which is a function of the two measurements derived from the extremal POVMs $M_i$, where $i\in\{1,2\}$ and $(\theta_{ia_i}, \phi_{ia_i})$ are spherical polar coordinates of the Bloch vector $\hat{m}_{ia_i}$. The weight parameters of the three-outcome (corresponding to indices $\kappa \in\{0,1,2\}$) POVMs, can be written by
\begin{eqnarray} 
    r_{i\kappa}\!&=&\!\frac{2 \sin \alpha_{i\kappa}}{\sin \alpha_{i0}\!+\!\sin \alpha_{i1}\!+\!\sin \alpha_{i2}},\label{eq:length-r}\\[3pt]
     &~&\alpha_{i0}+\alpha_{i1}+\alpha_{i2}=2\pi,\label{eq:angles}  \\[3pt]
    &~& 0\leq \alpha_{i0}, ~\alpha_{i1}, ~\alpha_{i2} \leq \pi \label{eq:angleconstraints}.      
\end{eqnarray}
See Fig.~\ref{3outPOVM} to visualize the measurements.\\

	\begin{figure}[t!]
		\begin{center}
			\includegraphics[scale=0.14]{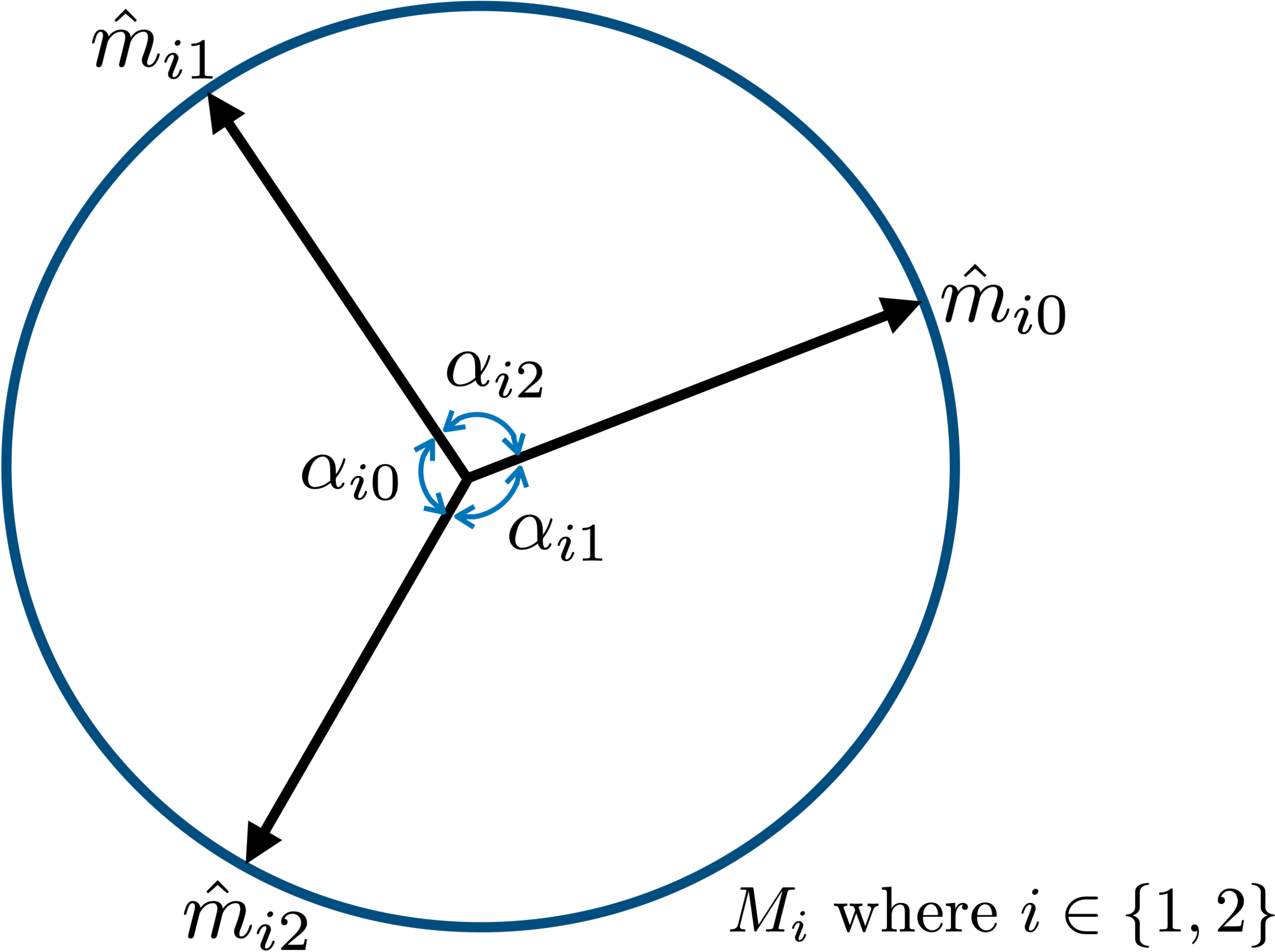}
			\caption{Illustration of a three-outcome extremal POVM (on a qubit) on a plane passing through the origin of the Bloch sphere. }
			\label{3outPOVM}
		\end{center}
	\end{figure}
    
    In general, the Bloch vectors corresponding to two different measurements may occupy two different planes. We choose, without loss of generality, the Bloch plane of the measurement  $M_1$ as the $XY$-plane. We choose the $Y$-axis along the line where the planes of the two measurements $M_1$ and $M_2$ intersect.  Let us denote the angle of the unit vector normal to the plane of the measurement $M_2$ with the $Z$-axis by $0\leq \Theta\leq \pi$. Further, without loss of generality, for the second measurement $M_2$ we set the angle of the Bloch vector of the first outcome $\phi_{20}=0$, and for the first measurement let us denote the angle of the Bloch vector of the first outcome $\phi_{10}=\Phi$. Now, the objective function in Eq.~(\ref{eq:f}) can be re-written in terms of the variables $\alpha_{ia_i}$, $\Theta$, and $\Phi$ by applying Eq.~(\ref{eq:length-r}) and transforming the variables $(\theta_{ia_i}, \phi_{ia_i})$ to the new variables. \\

Let us then choose a vector orthonormal to the plane $M_2$ as $\hat{\mathbf{n}}_2=(-\sin\Theta,~0,~\cos\Theta)$ and a reference vector in the plane $M_2$ as $\hat{\mathbf{p}}_2=(\sin(\pi/2-\Theta),~0,~\cos(\pi/2-\Theta))$. The Bloch vectors of the POVM elements $M_{2j}$, where $j\in \{0,1,2\}$ is given by
\begin{equation}
    \hat{m}_{2j}=(\sin\theta_{2j}\cos\phi_{2j},~\sin\theta_{2j}\sin\phi_{2j},~\cos\theta_{2j})
\end{equation}
One can solve the system of equations
\begin{eqnarray}
    \hat{\mathbf{n}}_2\cdot \hat{m}_{2j}&=&0  \\
     \hat{\mathbf{p}}_2\cdot \hat{m}_{2j}&=&\cos \alpha_{2j}  
\end{eqnarray}
to find the variables $\theta_{2j}$ and $\phi_{2j}$ in terms of the variable $\Theta$ and $\alpha_{2j}$. Thus, the objective function $f$ in Eq.~(\ref{eq:f})
is now written in terms of six free variables: i.e., $f(\Theta, \Phi, \alpha_{10}, \alpha_{12}, \alpha_{20},\alpha_{22})$ (here we have applied Eq.~(\ref{eq:angles}) to eliminate the two dependent variables $\alpha_{11}$ and $\alpha_{21}$). Then the following proposition holds.\\

\begin{proposition}\label{prop:reduction-to-XY-plane}
   The optimal value of $f$ in Eq.~(\ref{eq:f}) is achieved with both the three outcome POVMs $M_1$ and $M_2$ in the $XY$-plane of the Bloch sphere, that is, when variable $\Theta =0$.
\end{proposition}  
\begin{proof}---We omit the lengthy expressions and steps, but the idea is simple. The objective function in Eq.~(\ref{eq:f}) is written in terms of $\Theta$ and the remaining variables. Then, we analyze the effect of changing the angle $\Theta$ between the two measurement planes. This is done by taking the partial derivative of the objective function $f$ with respect to the variable $\Theta$. We find that for an arbitrarily fixed value of all the other variables, the partial derivative with respect to $\Theta$ is either always nonnegative or always nonpositive, i.e., the function is either increasing or decreasing. A maximum (for any fixed value of all the other variables) will therefore be achieved at an end point in the interval $0\leq \Theta\leq \pi$, implying that the optimal measurement $M_2$ lies in the $XY$-plane of the Bloch sphere (i.e., both the measurements are in the same plane). 
\end{proof}
\vspace{2cm}
Applying the above proposition, the optimization problem is now reduced to maximizing
\begin{eqnarray} 
   \!\!\!\!\!\!f^{\ast}\!&=&\!\!\sum_{a_2\!=\!0}^{2}\!\sum_{a_1\!=\!0}^{2}\!\!\left\{\!r_{1a_1}^2\!\!+\!r_{2a_2}^2\!\!+\!2r_{1a_1}\!r_{2a_2}\!\cos(\phi_{1a_1}\!\!-\!\phi_{2a_2})\right\}^{\frac{1}{2}\!}, \nonumber \\ \label{eq:fstar}
\end{eqnarray}
where both the three outcome measurements are from the $XY$-plane.
In Fig.~(\ref{xy-plane}) two three outcome extremal POVMs in the XY-plane are shown.
	\begin{figure}[t!]
		\begin{center}
			\includegraphics[scale=0.15]{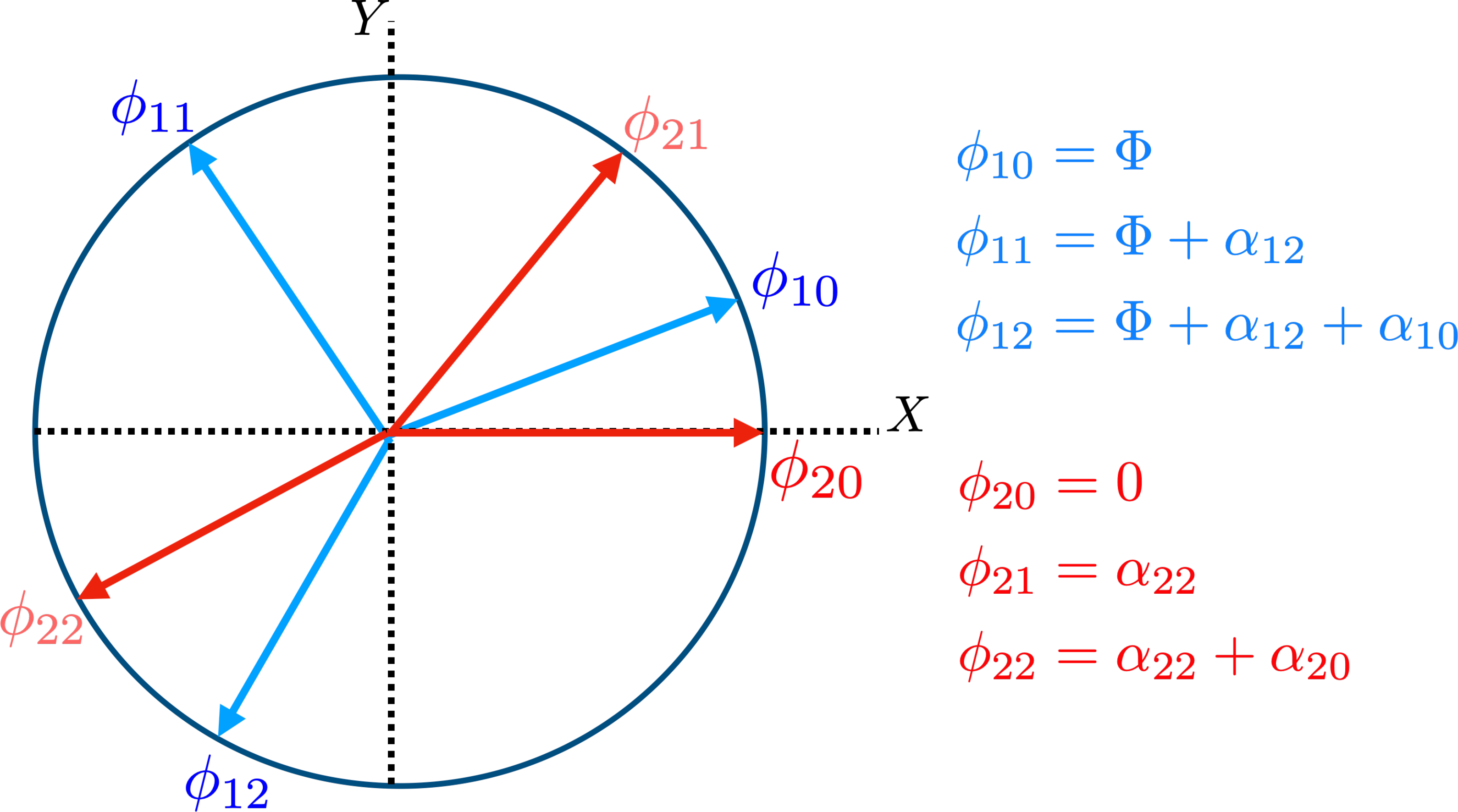}	
   \caption{Illustration of three outcome extremal POVMs on the $XY$-plane of the Bloch sphere.} \label{xy-plane}
		\end{center}
	\end{figure}
In Eq.~(\ref{eq:fstar}), for $i\in\{1,2\}$, we substitute
\begin{eqnarray}
 r_{i0}&=& \frac{2 \sin \alpha_{i0}}{\sin \alpha_{i0}-\sin (\alpha_{i0}+\alpha_{i2})+\sin \alpha_{i2}},\\[5pt]
     r_{i1}&=&\frac{-2\sin (\alpha_{i0}+\alpha_{i2})}{\sin \alpha_{i0}-\sin (\alpha_{i0}+\alpha_{i2})+\sin \alpha_{i2}},\\[5pt]
      r_{i2}&=& \frac{2\sin \alpha_{i2}}{\sin \alpha_{i0}-\sin (\alpha_{i0}+\alpha_{i2})+\sin \alpha_{i2}}; \\[10pt]
    \phi_{i0}&=&(2-i)~\Phi,\\
    \phi_{i1}&=& (2-i)~\Phi + \alpha_{i2},\\
    \phi_{i2}&=& (2-i)~\Phi + \alpha_{i2} + \alpha_{i0}.    
\end{eqnarray}
The free variables in the optimization problem are now explicit; constraints on these variables are given by
\begin{eqnarray}
    &~&0\leq \alpha_{i0},~\alpha_{i2}\leq \pi ~\mbox{and}~\pi \leq \alpha_{i0} + \alpha_{i2} \leq 2 \pi,\\
    &~&0\leq \Phi \leq \pi.
\end{eqnarray}

\begin{proposition}\label{prop:qmax}
   The optimal value of $f^{\ast}$ in Eq.~(\ref{eq:fstar}) is given by $4(1+\sqrt{2})$ and is achieved when both measurements $M_1$ and $M_2$ are two outcome projective measurements in orthogonal directions. Consequently, the optimal quantum success probability for the REC task is $(7+\sqrt{2})/9$ (and that for the corresponding RAC task is $(4+\sqrt{2})/9$); in particular, the quantum protocol presented in the main text is an optimal one for the REC task.
\end{proposition}  

Our proposition is supported by our analytical simplification (reduction) of the optimization problem stated in proposition~(\ref{prop:reduction-to-XY-plane}), followed by numerical optimization implemented on Mathematica, giving an approximate maximum value $\max(f^{\ast})\approx 9.65685\approx 4(1+\sqrt{2})$.\\


    \section{Bounds on random exclusion task for alphabet size m} \label{appendB}
    \begin{proposition}
        The maximum classical average success probability for the random exclusion task with alphabet size $m$ is $\max_C ~{p}^{\mathrm{(REC)}}=(1-\frac{1}{m^2})$.
    \end{proposition}
\begin{proof}
Let us first find the optimal value over all the deterministic protocols. Our proof follows closely that of the $(2,3)$ REC analyses in the main text. The two stages of any deterministic protocol are as follows. (i) Alice sending one bit information $c\in\{0,1\}$ corresponds to some two partition $\mathbb{P}_0, \mathbb{P}_1$ of the set of all possible words $\mathcal{A}^2:=\{a_1a_2: a_1,a_2\in\{0,1,2,\dots,m-1\}\}$; if the received word $a=a_1a_2\in \mathbb{P}_0$, Alice communicates $c=0$ to Bob, otherwise she sends $c=1$. (ii) Bob's answer depends on the received bit, he answers with, say, $b_i^c\in \{0,1,2,\dots , m-1\}$, i.e., Bob's deterministic strategy corresponds to assigning values to $b_1^0,b_2^0,~b_1^1,b_2^1 \in \{0,1,2,\dots,m-1\}$.\\

First we show that for any deterministic strategy of Alice and Bob, for at least two out of the total $2m^2$ possible instances ($m^2$ possible words to Alice and two possible questions to Bob), Bob's answers will be incorrect. To show this, first, we note that from any deterministic strategy of Bob one can derive the following two words $b_1^0b_2^0, b_1^1b_2^1 \in \mathcal{A}^2$, therefore, for these two words, given a deterministic strategy of Alice, there are four possibilities:
\begin{enumerate}
    \item $b_1^0b_2^0, ~b_1^1b_2^1 \in \mathbb{P}_0$,
    \item $b_1^0b_2^0, ~b_1^1b_2^1 \in \mathbb{P}_1$,
    \item $b_1^0b_2^0 \in \mathbb{P}_0$ and $b_1^1b_2^1 \in \mathbb{P}_1$,
    \item $b_1^0b_2^0 \in \mathbb{P}_1$ and $ b_1^1b_2^1 \in \mathbb{P}_0$.
    
\end{enumerate}
In the case-1, when Alice receives the word $b_1^0b_2^0$, Bob's answer is $b_i^0$ which is incorrect at both the positions $i\in\{1,2\}$. In case-2, similar to case-1, when Alice receives the word $b_1^1b_2^1$, Bob's answer is $b_i^1$ which is incorrect at both the positions $i\in\{1,2\}$.\\

In case-3 and case-4 it can never happen that $b_1^0b_2^0=b_1^1b_2^1$ because same word cannot appear in two distinct partitions.\\

For the case-3, when Alice receives the word $b_1^0b_2^0$ ($b_1^1b_2^1$), Bob's answer is $b_i^0$ ($b_i^1$) which is incorrect at both the positions $i\in\{1,2\}$; thus in the case-3 Bob answers incorrectly in at least four instances, namely $b_1^0b_2^0,b_1^1b_2^1$.\\

Lastly, in the case-4, first suppose that $b_1^0b_2^1 \neq b_2^0b_1^1$. Then for the word $b_1^0b_2^1$, to whichever partition it belongs, Bob's answer will be incorrect at least at one position $i\in\{1,2\}$; similarly, for the word $b_2^0b_1^1$, to whichever partition it belongs, Bob's answer will be incorrect at least at one position $i\in\{1,2\}$. In the remaining situation, i.e., $b_1^0b_2^1 = b_2^0b_1^1\Leftrightarrow b_1^0=b_2^0~\mbox{and}~b_1^1=b_2^1$, on considering the words $b_1^0b_1^1$ and $b_1^1b_1^0$ (irrespective which partition they belong to) Bob's answer is incorrect at least two times; at least once for the first word $b_1^0b_1^1$ and at least once for the second word $b_1^1b_1^0$. Then, in case-4 as well, Bob's answer is incorrect in least in two instances.\\

Therefore, in all the four cases, we get that for any deterministic strategy of Alice and Bob, in at least two out of $2m^2$ cases Bob answers incorrectly. Therefore, maximal average success probability with deterministic strategies is upper bounded by $\frac{2m^2-2}{2m^2}=1-\frac{1}{m^2}$.\newline

Finally, we show that probabilistic strategies cannot do better. Lets say there are total $N$ number of deterministic strategies $S_j$ where $j\in\{1,\cdots,N\}$. Suppose, Alice and Bob have shared random variables $\{\lambda_j\}_{j=1}^{N}$ with probability distribution $p(\lambda_j)=p_j$. Note that any probabilistic strategy can be implemented with an access to the considered shared random variables. Then if Alice and Bob apply strategy $S_j$ with probability $p_j$ we get a probabilistic strategy $S=\{p_j,S_j\}_{j=1}^N$. Then the average success probability with any such strategy $S$ is ${p}^{\mathrm{(REC)}}_{S}= \sum_j p_j {p}^{\mathrm{(REC)}}_{S_j} \leq  \sum_j p_j \left(1-\frac{1}{m^2} \right)= 1-\frac{1}{m^2}$. This completes the proof.
\end{proof}

{\bf Quantum advantage.---} Here we will show that there is a quantum protocol which considers only two outcome measurements for decoding to beat the optimal classical value derived above. Alice sends the information on her words $a_1a_2$ by encoding in a one qubit state $\rho_{a_1a_2}$ as follows:

\begin{widetext}
\begin{align}
	\rho_{00}&= \frac{1}{2}\left(\mathbb{I}+\sigma_x\right),&\rho_{01}&= \frac{1}{2}\left(\mathbb{I}+\sigma_z\right),&\rho_{0c_2}&=\frac{1}{2}\left(\mathbb{I}+\frac{\sigma_x+\sigma_z}{\sqrt{2}}\right) \nonumber \\[10pt]
	\rho_{10}&= \frac{1}{2}\left(\mathbb{I}-\sigma_z\right),&\rho_{11}&= \frac{1}{2}\left(\mathbb{I}-\sigma_x\right),&\rho_{1c_2}&= \frac{1}{2}\left(\mathbb{I}-\frac{\sigma_x+\sigma_z}{\sqrt{2}}\right) \nonumber \\[10pt]
	\rho_{c_10}&= \frac{1}{2}\left(\mathbb{I}+\frac{\sigma_x-\sigma_z}{\sqrt{2}}\right),&\rho_{c_1 1}&= \frac{1}{2}\left(\mathbb{I}+\frac{\sigma_z-\sigma_x}{\sqrt{2}}\right),&\rho_{c_1 c_2}&= \frac{\mathbb{I}}{2} \nonumber \\[10pt]
    \text{where } c_1,c_2 &\in \lbrace 2,3,\ldots,m-1 \rbrace.
\end{align}
All the qubit states used for encoding are from the $XZ$-plane of the Bloch sphere. To answer the first position in the Alice's word, Bob performs a two outcome projective measurement $M_{1}$; on measurement outcome $+1$ ($-1$) he answers $0$ ($1$). Whereas for answering the second position he measures $M_2$ and answers with $0$ ($1$) when the measurement outcome is $+1$ ($-1$). The average success probability of our protocol is
\begin{eqnarray}
	{p}^{\mathrm{(REC)}}_{Q}&:=&1-\frac{1}{2m^2} \sum_{a_1,a_2\in\mathcal{A}} ~\mbox{Tr}\left[\rho_{a_1a_2} \left(M_{1a_1}+M_{2a_2}\right)\right]\nonumber \\[10pt]
	&=& 1-\frac{2-\sqrt{2}}{m^2}.\nonumber
\end{eqnarray}
Here the first and second measurements are implemented as follows
\begin{eqnarray}
	M_1&\equiv& \left\{M_{10}=\frac{1}{2}\left(\mathbb{I}-\frac{\sigma_x+\sigma_z}{\sqrt{2}}\right),~M_{11}=\frac{1}{2}\left(\mathbb{I}+\frac{\sigma_x+\sigma_z}{\sqrt{2}}\right),~M_{1c_1}=0\right\}, \nonumber\\
		M_2&\equiv& \left\{M_{20}=\frac{1}{2}\left(\mathbb{I}+\frac{\sigma_z-\sigma_x}{\sqrt{2}}\right),~M_{21}=\frac{1}{2}\left(\mathbb{I}+\frac{\sigma_x-\sigma_z}{\sqrt{2}}\right),~M_{2c_2}=0\right\}.
\end{eqnarray} 
\end{widetext}
Note that Bob never answers $\lbrace 2,3,\ldots,m-1 \rbrace$, still our protocol gives a high success because the goal here is to answer with alphabet different from those appearing in the words. Also note that the state $\rho_{c_1 c_2}$ used for encoding the words $c_1 \; c_2$, can be any arbitrary state, even any mixed state, e.g., maximally mixed state $\frac{\mathbb{I}}{2}$. This is due to the fact $M_{1c_1}=0$ and $M_{2c_2}=0$, i.e., outcomes $c_1 \text{ or } c_2$ never occur for the considered measurements.\\

\begin{proposition}
    The maximum average success probability in the quantum random exclusion task, with alphabet size $m$, for a two-outcome POVM is achieved with the above discussed quantum protocol.
\end{proposition}

\begin{proof}
Let us address the scenario wherein Alice sends a qubit instead of a bit. We can show that there is a quantum advantage in this case for all $m$. The proof of advantage only requires Bob to perform two 2-outcome measurements. Say we form the partition as $\mathbb{Q}=\lbrace 00,01,10,11 \rbrace$ and $\overline{\mathbb{Q}}=\lbrace a_1a_2\in\mathcal{A}^2: a_1a_2\notin \mathbb{Q} \rbrace$. The two 2-outcome measurements can be used to distinguish between states associated with $\mathbb{Q}$. For the states associated with $\overline{\mathbb{Q}}$, the measurement output is always correct for our task. We can find an upperbound for the quantum success probability as well: 
\begin{widetext}
\begin{align*}
    {p}^{\mathrm{(REC)}}_{Q} &= \frac{1}{2m^2} \sum_{a\in\mathcal{A}^2}\sum_{i=1}^2{\rm Tr} \left[ (\mathbb{I}-M_{ia_i}) \rho_{a_1 a_2}  \right] \\ 
    &= \frac{4}{m^2}\left\{\frac{1}{8}\sum \limits_{a\in \mathbb{Q}} \sum\limits_{i=1}^2 {\rm Tr} \left[(\mathbb{I}-M_{ia_i}) \rho_{a_1 a_2}\right]\right\}+\frac{1}{2m^2}\sum \limits_{a\in \overline{\mathbb{Q}}}\left\{\sum\limits_{i=1}^2 {\rm Tr} [ (\mathbb{I}-M_{ia_i}) \rho_{a_1 a_2}  ] \right\} \\
    &\leq \frac{4}{m^2}\left\{\frac{1}{8}\sum \limits_{a\in \mathbb{Q}} \sum\limits_{i=1}^2 {\rm Tr} \left[(\mathbb{I}-M_{ia_i}) \rho_{a_1 a_2}\right]\right\} +\frac{m^2-4}{m^2} \\
     &\leq \frac{4}{m^2}\left\{\max\limits_{\{M_{ia_i},~\rho_{a_1a_2}\}}~\frac{1}{8}\sum \limits_{a\in \mathbb{Q}} \sum\limits_{i=1}^2 {\rm Tr} \left[(\mathbb{I}-M_{ia_i}) \rho_{a_1 a_2}\right]\right\} +\frac{m^2-4}{m^2}\\
      &= \frac{4}{m^2}\left\{\max\limits_{\{\tilde{M}_{ia_i},~\rho_{a_1a_2}\}}~\frac{1}{8}\sum \limits_{a\in \mathbb{Q}} \sum\limits_{i=1}^2 {\rm Tr} \left[\tilde{M}_{ia_i}~ \rho_{a_1 a_2}\right]\right\} +\frac{m^2-4}{m^2} ~~~~(~\mbox{here}~\tilde{M}_{ia_i}=\mathbb{I}-M_{ia_i}~),\\
      &\leq \frac{4}{m^2}\left\{\frac{1}{2}\left(1+\frac{1}{\sqrt{2}}\right)\right\}+\frac{m^2-4}{m^2} \\ &= 1-\frac{2-\sqrt{2}}{m^2}.
\end{align*}
\end{widetext}
In the last step above, we simply apply the known result, i.e., in a $2\rightarrow 1$ random access codes the quantum maximum for the average success probability is $(1+1/\sqrt{2})/2$. Since our quantum protocol reaches the derived quantum upper bound, we find an optimal over the set of all two-outcome measurement protocols.
\end{proof}
With this result, we can see that, as long as Alice is only using a two level system to communicate to Bob, Alice can achieve a quantum advantage for any finite alphabet size $m$ of Alice's input since $1-\frac{2-\sqrt{2}}{m^2} > 1-\frac{1}{m^2}$.

\bibliography{references.bib}

\begin{thebibliography}{41}%
\makeatletter
\providecommand \@ifxundefined [1]{%
 \@ifx{#1\undefined}
}%
\providecommand \@ifnum [1]{%
 \ifnum #1\expandafter \@firstoftwo
 \else \expandafter \@secondoftwo
 \fi
}%
\providecommand \@ifx [1]{%
 \ifx #1\expandafter \@firstoftwo
 \else \expandafter \@secondoftwo
 \fi
}%
\providecommand \natexlab [1]{#1}%
\providecommand \enquote  [1]{``#1''}%
\providecommand \bibnamefont  [1]{#1}%
\providecommand \bibfnamefont [1]{#1}%
\providecommand \citenamefont [1]{#1}%
\providecommand \href@noop [0]{\@secondoftwo}%
\providecommand \href [0]{\begingroup \@sanitize@url \@href}%
\providecommand \@href[1]{\@@startlink{#1}\@@href}%
\providecommand \@@href[1]{\endgroup#1\@@endlink}%
\providecommand \@sanitize@url [0]{\catcode `\\12\catcode `\$12\catcode
  `\&12\catcode `\#12\catcode `\^12\catcode `\_12\catcode `\%12\relax}%
\providecommand \@@startlink[1]{}%
\providecommand \@@endlink[0]{}%
\providecommand \url  [0]{\begingroup\@sanitize@url \@url }%
\providecommand \@url [1]{\endgroup\@href {#1}{\urlprefix }}%
\providecommand \urlprefix  [0]{URL }%
\providecommand \Eprint [0]{\href }%
\providecommand \doibase [0]{http://dx.doi.org/}%
\providecommand \selectlanguage [0]{\@gobble}%
\providecommand \bibinfo  [0]{\@secondoftwo}%
\providecommand \bibfield  [0]{\@secondoftwo}%
\providecommand \translation [1]{[#1]}%
\providecommand \BibitemOpen [0]{}%
\providecommand \bibitemStop [0]{}%
\providecommand \bibitemNoStop [0]{.\EOS\space}%
\providecommand \EOS [0]{\spacefactor3000\relax}%
\providecommand \BibitemShut  [1]{\csname bibitem#1\endcsname}%
\let\auto@bib@innerbib\@empty
\bibitem [{\citenamefont {Winter}(2013)}]{winter}%
  \BibitemOpen
  \bibfield  {author} {\bibinfo {author} {\bibfnamefont {A.}~\bibnamefont
  {Winter}},\ }\bibfield  {title} {\enquote {\bibinfo {title} {{W}inter's
  {P}rinciple, {B}enasque},}\ }\href@noop {} {\  (\bibinfo {year}
  {2013})}\BibitemShut {NoStop}%
\bibitem [{\citenamefont {Holevo}(2011)}]{Holevo2011}%
  \BibitemOpen
  \bibfield  {author} {\bibinfo {author} {\bibfnamefont {Alexander~S}\
  \bibnamefont {Holevo}},\ }\href@noop {} {\emph {\bibinfo {title}
  {Probabilistic and statistical aspects of quantum theory}}},\ Publications of
  the Scuola Normale Superiore\ (\bibinfo  {publisher} {Scuola Normale
  Superiore},\ \bibinfo {address} {Pisa, Italy},\ \bibinfo {year}
  {2011})\BibitemShut {NoStop}%
\bibitem [{\citenamefont {Holevo}(1998)}]{651037}%
  \BibitemOpen
  \bibfield  {author} {\bibinfo {author} {\bibfnamefont {A.S.}\ \bibnamefont
  {Holevo}},\ }\bibfield  {title} {\enquote {\bibinfo {title} {The capacity of
  the quantum channel with general signal states},}\ }\href {\doibase
  10.1109/18.651037} {\bibfield  {journal} {\bibinfo  {journal} {IEEE
  Transactions on Information Theory}\ }\textbf {\bibinfo {volume} {44}},\
  \bibinfo {pages} {269--273} (\bibinfo {year} {1998})}\BibitemShut {NoStop}%
\bibitem [{\citenamefont {Schumacher}\ and\ \citenamefont
  {Westmoreland}(1997)}]{PhysRevA.56.131}%
  \BibitemOpen
  \bibfield  {author} {\bibinfo {author} {\bibfnamefont {Benjamin}\
  \bibnamefont {Schumacher}}\ and\ \bibinfo {author} {\bibfnamefont
  {Michael~D.}\ \bibnamefont {Westmoreland}},\ }\bibfield  {title} {\enquote
  {\bibinfo {title} {Sending classical information via noisy quantum
  channels},}\ }\href {\doibase 10.1103/PhysRevA.56.131} {\bibfield  {journal}
  {\bibinfo  {journal} {Phys. Rev. A}\ }\textbf {\bibinfo {volume} {56}},\
  \bibinfo {pages} {131--138} (\bibinfo {year} {1997})}\BibitemShut {NoStop}%
\bibitem [{\citenamefont {Ambainis}\ \emph {et~al.}(2009)\citenamefont
  {Ambainis}, \citenamefont {Leung}, \citenamefont {Mancinska},\ and\
  \citenamefont {Ozols}}]{Ambainis2009}%
  \BibitemOpen
  \bibfield  {author} {\bibinfo {author} {\bibfnamefont {Andris}\ \bibnamefont
  {Ambainis}}, \bibinfo {author} {\bibfnamefont {Debbie}\ \bibnamefont
  {Leung}}, \bibinfo {author} {\bibfnamefont {Laura}\ \bibnamefont
  {Mancinska}}, \ and\ \bibinfo {author} {\bibfnamefont {Maris}\ \bibnamefont
  {Ozols}},\ }\bibfield  {title} {\enquote {\bibinfo {title} {Quantum random
  access codes with shared randomness},}\ }\href@noop {} {\  (\bibinfo {year}
  {2009})},\ \Eprint {http://arxiv.org/abs/0810.2937} {arXiv:0810.2937
  [quant-ph]} \BibitemShut {NoStop}%
\bibitem [{\citenamefont {Ambainis}\ \emph {et~al.}(2024)\citenamefont
  {Ambainis}, \citenamefont {Kravchenko}, \citenamefont {Sazim}, \citenamefont
  {Bae},\ and\ \citenamefont {Rai}}]{Ambainis2024RAC}%
  \BibitemOpen
  \bibfield  {author} {\bibinfo {author} {\bibfnamefont {Andris}\ \bibnamefont
  {Ambainis}}, \bibinfo {author} {\bibfnamefont {Dmitry}\ \bibnamefont
  {Kravchenko}}, \bibinfo {author} {\bibfnamefont {Sk}~\bibnamefont {Sazim}},
  \bibinfo {author} {\bibfnamefont {Joonwoo}\ \bibnamefont {Bae}}, \ and\
  \bibinfo {author} {\bibfnamefont {Ashutosh}\ \bibnamefont {Rai}},\ }\bibfield
   {title} {\enquote {\bibinfo {title} {Quantum advantages in $(n,d)\mapsto 1$
  random access codes},}\ }\href {\doibase 10.1088/1367-2630/ad9bdf} {\bibfield
   {journal} {\bibinfo  {journal} {New J. Phys.}\ }\textbf {\bibinfo {volume}
  {26}},\ \bibinfo {pages} {123023} (\bibinfo {year} {2024})}\BibitemShut
  {NoStop}%
\bibitem [{\citenamefont {Tavakoli}\ \emph {et~al.}(2015)\citenamefont
  {Tavakoli}, \citenamefont {Hameedi}, \citenamefont {Marques},\ and\
  \citenamefont {Bourennane}}]{Tavakoli2015}%
  \BibitemOpen
  \bibfield  {author} {\bibinfo {author} {\bibfnamefont {Armin}\ \bibnamefont
  {Tavakoli}}, \bibinfo {author} {\bibfnamefont {Alley}\ \bibnamefont
  {Hameedi}}, \bibinfo {author} {\bibfnamefont {Breno}\ \bibnamefont
  {Marques}}, \ and\ \bibinfo {author} {\bibfnamefont {Mohamed}\ \bibnamefont
  {Bourennane}},\ }\bibfield  {title} {\enquote {\bibinfo {title} {Quantum
  random access codes using single $d$-level systems},}\ }\href {\doibase
  10.1103/PhysRevLett.114.170502} {\bibfield  {journal} {\bibinfo  {journal}
  {Phys. Rev. Lett.}\ }\textbf {\bibinfo {volume} {114}},\ \bibinfo {pages}
  {170502} (\bibinfo {year} {2015})}\BibitemShut {NoStop}%
\bibitem [{\citenamefont {Carmeli}\ \emph {et~al.}(2020)\citenamefont
  {Carmeli}, \citenamefont {Heinosaari},\ and\ \citenamefont
  {Toigo}}]{carmeli2020quantum}%
  \BibitemOpen
  \bibfield  {author} {\bibinfo {author} {\bibfnamefont {Claudio}\ \bibnamefont
  {Carmeli}}, \bibinfo {author} {\bibfnamefont {Teiko}\ \bibnamefont
  {Heinosaari}}, \ and\ \bibinfo {author} {\bibfnamefont {Alessandro}\
  \bibnamefont {Toigo}},\ }\bibfield  {title} {\enquote {\bibinfo {title}
  {Quantum random access codes and incompatibility of measurements},}\ }\href
  {\doibase 10.1209/0295-5075/130/50001} {\bibfield  {journal} {\bibinfo
  {journal} {Europhysics Lett.}\ }\textbf {\bibinfo {volume} {130}},\ \bibinfo
  {pages} {50001} (\bibinfo {year} {2020})}\BibitemShut {NoStop}%
\bibitem [{\citenamefont {Helstrom}(1969)}]{Helstrom:1969aa}%
  \BibitemOpen
  \bibfield  {author} {\bibinfo {author} {\bibfnamefont {Carl~W.}\ \bibnamefont
  {Helstrom}},\ }\bibfield  {title} {\enquote {\bibinfo {title} {Quantum
  detection and estimation theory},}\ }\href {\doibase 10.1007/BF01007479}
  {\bibfield  {journal} {\bibinfo  {journal} {Journal of Statistical Physics}\
  }\textbf {\bibinfo {volume} {1}},\ \bibinfo {pages} {231--252} (\bibinfo
  {year} {1969})}\BibitemShut {NoStop}%
\bibitem [{\citenamefont {Bae}\ and\ \citenamefont {Kwek}(2015)}]{Bae_2015}%
  \BibitemOpen
  \bibfield  {author} {\bibinfo {author} {\bibfnamefont {Joonwoo}\ \bibnamefont
  {Bae}}\ and\ \bibinfo {author} {\bibfnamefont {Leong-Chuan}\ \bibnamefont
  {Kwek}},\ }\bibfield  {title} {\enquote {\bibinfo {title} {Quantum state
  discrimination and its applications},}\ }\href {\doibase
  10.1088/1751-8113/48/8/083001} {\bibfield  {journal} {\bibinfo  {journal} {J.
  Phys. A: Math. Theor.}\ }\textbf {\bibinfo {volume} {48}},\ \bibinfo {pages}
  {083001} (\bibinfo {year} {2015})}\BibitemShut {NoStop}%
\bibitem [{\citenamefont {Barnett}\ and\ \citenamefont
  {Croke}(2009)}]{Barnett:09}%
  \BibitemOpen
  \bibfield  {author} {\bibinfo {author} {\bibfnamefont {Stephen~M.}\
  \bibnamefont {Barnett}}\ and\ \bibinfo {author} {\bibfnamefont {Sarah}\
  \bibnamefont {Croke}},\ }\bibfield  {title} {\enquote {\bibinfo {title}
  {Quantum state discrimination},}\ }\href {\doibase 10.1364/AOP.1.000238}
  {\bibfield  {journal} {\bibinfo  {journal} {Adv. Opt. Photon.}\ }\textbf
  {\bibinfo {volume} {1}},\ \bibinfo {pages} {238--278} (\bibinfo {year}
  {2009})}\BibitemShut {NoStop}%
\bibitem [{\citenamefont {Bergou}(2007)}]{Bergou:2007aa}%
  \BibitemOpen
  \bibfield  {author} {\bibinfo {author} {\bibfnamefont {J{\'a}nos~A}\
  \bibnamefont {Bergou}},\ }\bibfield  {title} {\enquote {\bibinfo {title}
  {Quantum state discrimination and selected applications},}\ }\href {\doibase
  10.1088/1742-6596/84/1/012001} {\bibfield  {journal} {\bibinfo  {journal}
  {Journal of Physics: Conference Series}\ }\textbf {\bibinfo {volume} {84}},\
  \bibinfo {pages} {012001} (\bibinfo {year} {2007})}\BibitemShut {NoStop}%
\bibitem [{\citenamefont {Caves}\ \emph {et~al.}(2002)\citenamefont {Caves},
  \citenamefont {Fuchs},\ and\ \citenamefont {Schack}}]{Caves2002}%
  \BibitemOpen
  \bibfield  {author} {\bibinfo {author} {\bibfnamefont {Carlton~M.}\
  \bibnamefont {Caves}}, \bibinfo {author} {\bibfnamefont {Christopher~A.}\
  \bibnamefont {Fuchs}}, \ and\ \bibinfo {author} {\bibfnamefont {R\"{u}diger}\
  \bibnamefont {Schack}},\ }\bibfield  {title} {\enquote {\bibinfo {title}
  {Conditions for compatibility of quantum-state assignments},}\ }\href
  {http://dx.doi.org/10.1103/PhysRevA.66.062111} {\bibfield  {journal}
  {\bibinfo  {journal} {Phys. Rev. A}\ }\textbf {\bibinfo {volume} {66}}
  (\bibinfo {year} {2002})}\BibitemShut {NoStop}%
\bibitem [{\citenamefont {Pusey}\ \emph {et~al.}(2012)\citenamefont {Pusey},
  \citenamefont {Barrett},\ and\ \citenamefont {Rudolph}}]{PuseyPBR2012}%
  \BibitemOpen
  \bibfield  {author} {\bibinfo {author} {\bibfnamefont {Matthew}\ \bibnamefont
  {Pusey}}, \bibinfo {author} {\bibfnamefont {Jonathan}\ \bibnamefont
  {Barrett}}, \ and\ \bibinfo {author} {\bibfnamefont {Terry}\ \bibnamefont
  {Rudolph}},\ }\bibfield  {title} {\enquote {\bibinfo {title} {On the reality
  of the quantum state},}\ }\href {\doibase 10.1038/nphys2309} {\bibfield
  {journal} {\bibinfo  {journal} {Nature Physics}\ }\textbf {\bibinfo {volume}
  {8}},\ \bibinfo {pages} {475} (\bibinfo {year} {2012})}\BibitemShut {NoStop}%
\bibitem [{\citenamefont {Bandyopadhyay}\ \emph
  {et~al.}(2014{\natexlab{a}})\citenamefont {Bandyopadhyay}, \citenamefont
  {Jain}, \citenamefont {Oppenheim},\ and\ \citenamefont
  {Perry}}]{PhysRevA.89.022336}%
  \BibitemOpen
  \bibfield  {author} {\bibinfo {author} {\bibfnamefont {Somshubhro}\
  \bibnamefont {Bandyopadhyay}}, \bibinfo {author} {\bibfnamefont {Rahul}\
  \bibnamefont {Jain}}, \bibinfo {author} {\bibfnamefont {Jonathan}\
  \bibnamefont {Oppenheim}}, \ and\ \bibinfo {author} {\bibfnamefont
  {Christopher}\ \bibnamefont {Perry}},\ }\bibfield  {title} {\enquote
  {\bibinfo {title} {Conclusive exclusion of quantum states},}\ }\href
  {\doibase 10.1103/PhysRevA.89.022336} {\bibfield  {journal} {\bibinfo
  {journal} {Phys. Rev. A}\ }\textbf {\bibinfo {volume} {89}},\ \bibinfo
  {pages} {022336} (\bibinfo {year} {2014}{\natexlab{a}})}\BibitemShut
  {NoStop}%
\bibitem [{\citenamefont {Heinosaari}\ and\ \citenamefont
  {Kerppo}(2018)}]{HeKe18}%
  \BibitemOpen
  \bibfield  {author} {\bibinfo {author} {\bibfnamefont {T.}~\bibnamefont
  {Heinosaari}}\ and\ \bibinfo {author} {\bibfnamefont {O.}~\bibnamefont
  {Kerppo}},\ }\bibfield  {title} {\enquote {\bibinfo {title}
  {Antidistinguishability of pure quantum states},}\ }\href {\doibase
  10.1088/1751-8121/aad1fc} {\bibfield  {journal} {\bibinfo  {journal} {J.
  Phys. A: Math. Theor.}\ }\textbf {\bibinfo {volume} {51}},\ \bibinfo {pages}
  {365303} (\bibinfo {year} {2018})}\BibitemShut {NoStop}%
\bibitem [{\citenamefont {Uola}\ \emph {et~al.}(2020)\citenamefont {Uola},
  \citenamefont {Bullock}, \citenamefont {Kraft}, \citenamefont
  {Pellonp\"a\"a},\ and\ \citenamefont {Brunner}}]{PhysRevLett.125.110402}%
  \BibitemOpen
  \bibfield  {author} {\bibinfo {author} {\bibfnamefont {Roope}\ \bibnamefont
  {Uola}}, \bibinfo {author} {\bibfnamefont {Tom}\ \bibnamefont {Bullock}},
  \bibinfo {author} {\bibfnamefont {Tristan}\ \bibnamefont {Kraft}}, \bibinfo
  {author} {\bibfnamefont {Juha-Pekka}\ \bibnamefont {Pellonp\"a\"a}}, \ and\
  \bibinfo {author} {\bibfnamefont {Nicolas}\ \bibnamefont {Brunner}},\
  }\bibfield  {title} {\enquote {\bibinfo {title} {All quantum resources
  provide an advantage in exclusion tasks},}\ }\href {\doibase
  10.1103/PhysRevLett.125.110402} {\bibfield  {journal} {\bibinfo  {journal}
  {Phys. Rev. Lett.}\ }\textbf {\bibinfo {volume} {125}},\ \bibinfo {pages}
  {110402} (\bibinfo {year} {2020})}\BibitemShut {NoStop}%
\bibitem [{\citenamefont {Ducuara}\ \emph {et~al.}(2020)\citenamefont
  {Ducuara}, \citenamefont {Lipka-Bartosik},\ and\ \citenamefont
  {Skrzypczyk}}]{PhysRevResearch.2.033374}%
  \BibitemOpen
  \bibfield  {author} {\bibinfo {author} {\bibfnamefont {Andr\'es~F.}\
  \bibnamefont {Ducuara}}, \bibinfo {author} {\bibfnamefont {Patryk}\
  \bibnamefont {Lipka-Bartosik}}, \ and\ \bibinfo {author} {\bibfnamefont
  {Paul}\ \bibnamefont {Skrzypczyk}},\ }\bibfield  {title} {\enquote {\bibinfo
  {title} {Multiobject operational tasks for convex quantum resource theories
  of state-measurement pairs},}\ }\href {\doibase
  10.1103/PhysRevResearch.2.033374} {\bibfield  {journal} {\bibinfo  {journal}
  {Phys. Rev. Res.}\ }\textbf {\bibinfo {volume} {2}},\ \bibinfo {pages}
  {033374} (\bibinfo {year} {2020})}\BibitemShut {NoStop}%
\bibitem [{\citenamefont {Harrigan}\ and\ \citenamefont
  {Spekkens}(2010)}]{Harrigan:2010aa}%
  \BibitemOpen
  \bibfield  {author} {\bibinfo {author} {\bibfnamefont {Nicholas}\
  \bibnamefont {Harrigan}}\ and\ \bibinfo {author} {\bibfnamefont {Robert~W.}\
  \bibnamefont {Spekkens}},\ }\bibfield  {title} {\enquote {\bibinfo {title}
  {Einstein, incompleteness, and the epistemic view of quantum states},}\
  }\href {\doibase 10.1007/s10701-009-9347-0} {\bibfield  {journal} {\bibinfo
  {journal} {Found. Phys.}\ }\textbf {\bibinfo {volume} {40}},\ \bibinfo
  {pages} {125--157} (\bibinfo {year} {2010})}\BibitemShut {NoStop}%
\bibitem [{\citenamefont {Bell}(1964)}]{PhysicsPhysiqueFizika.1.195}%
  \BibitemOpen
  \bibfield  {author} {\bibinfo {author} {\bibfnamefont {J.~S.}\ \bibnamefont
  {Bell}},\ }\bibfield  {title} {\enquote {\bibinfo {title} {On the einstein
  podolsky rosen paradox},}\ }\href {\doibase
  10.1103/PhysicsPhysiqueFizika.1.195} {\bibfield  {journal} {\bibinfo
  {journal} {Physics Physique Fizika}\ }\textbf {\bibinfo {volume} {1}},\
  \bibinfo {pages} {195--200} (\bibinfo {year} {1964})}\BibitemShut {NoStop}%
\bibitem [{\citenamefont {Brunner}\ \emph {et~al.}(2014)\citenamefont
  {Brunner}, \citenamefont {Cavalcanti}, \citenamefont {Pironio}, \citenamefont
  {Scarani},\ and\ \citenamefont {Wehner}}]{Brunner2014}%
  \BibitemOpen
  \bibfield  {author} {\bibinfo {author} {\bibfnamefont {Nicolas}\ \bibnamefont
  {Brunner}}, \bibinfo {author} {\bibfnamefont {Daniel}\ \bibnamefont
  {Cavalcanti}}, \bibinfo {author} {\bibfnamefont {Stefano}\ \bibnamefont
  {Pironio}}, \bibinfo {author} {\bibfnamefont {Valerio}\ \bibnamefont
  {Scarani}}, \ and\ \bibinfo {author} {\bibfnamefont {Stephanie}\ \bibnamefont
  {Wehner}},\ }\bibfield  {title} {\enquote {\bibinfo {title} {Bell
  nonlocality},}\ }\href {\doibase 10.1103/RevModPhys.86.419} {\bibfield
  {journal} {\bibinfo  {journal} {Rev. Mod. Phys.}\ }\textbf {\bibinfo {volume}
  {86}},\ \bibinfo {pages} {419--478} (\bibinfo {year} {2014})}\BibitemShut
  {NoStop}%
\bibitem [{\citenamefont {Pironio}\ \emph {et~al.}(2010)\citenamefont
  {Pironio}, \citenamefont {Ac{\'\i}n}, \citenamefont {Massar}, \citenamefont
  {de~la Giroday}, \citenamefont {Matsukevich}, \citenamefont {Maunz},
  \citenamefont {Olmschenk}, \citenamefont {Hayes}, \citenamefont {Luo},
  \citenamefont {Manning},\ and\ \citenamefont {Monroe}}]{Pironio:2010aa}%
  \BibitemOpen
  \bibfield  {author} {\bibinfo {author} {\bibfnamefont {S.}~\bibnamefont
  {Pironio}}, \bibinfo {author} {\bibfnamefont {A.}~\bibnamefont {Ac{\'\i}n}},
  \bibinfo {author} {\bibfnamefont {S.}~\bibnamefont {Massar}}, \bibinfo
  {author} {\bibfnamefont {A.~Boyer}\ \bibnamefont {de~la Giroday}}, \bibinfo
  {author} {\bibfnamefont {D.~N.}\ \bibnamefont {Matsukevich}}, \bibinfo
  {author} {\bibfnamefont {P.}~\bibnamefont {Maunz}}, \bibinfo {author}
  {\bibfnamefont {S.}~\bibnamefont {Olmschenk}}, \bibinfo {author}
  {\bibfnamefont {D.}~\bibnamefont {Hayes}}, \bibinfo {author} {\bibfnamefont
  {L.}~\bibnamefont {Luo}}, \bibinfo {author} {\bibfnamefont {T.~A.}\
  \bibnamefont {Manning}}, \ and\ \bibinfo {author} {\bibfnamefont
  {C.}~\bibnamefont {Monroe}},\ }\bibfield  {title} {\enquote {\bibinfo {title}
  {Random numbers certified by bell's theorem},}\ }\href {\doibase
  10.1038/nature09008} {\bibfield  {journal} {\bibinfo  {journal} {Nature}\
  }\textbf {\bibinfo {volume} {464}},\ \bibinfo {pages} {1021--1024} (\bibinfo
  {year} {2010})}\BibitemShut {NoStop}%
\bibitem [{\citenamefont {Ac\'{\i}n}\ \emph {et~al.}(2007)\citenamefont
  {Ac\'{\i}n}, \citenamefont {Brunner}, \citenamefont {Gisin}, \citenamefont
  {Massar}, \citenamefont {Pironio},\ and\ \citenamefont
  {Scarani}}]{PhysRevLett.98.230501}%
  \BibitemOpen
  \bibfield  {author} {\bibinfo {author} {\bibfnamefont {Antonio}\ \bibnamefont
  {Ac\'{\i}n}}, \bibinfo {author} {\bibfnamefont {Nicolas}\ \bibnamefont
  {Brunner}}, \bibinfo {author} {\bibfnamefont {Nicolas}\ \bibnamefont
  {Gisin}}, \bibinfo {author} {\bibfnamefont {Serge}\ \bibnamefont {Massar}},
  \bibinfo {author} {\bibfnamefont {Stefano}\ \bibnamefont {Pironio}}, \ and\
  \bibinfo {author} {\bibfnamefont {Valerio}\ \bibnamefont {Scarani}},\
  }\bibfield  {title} {\enquote {\bibinfo {title} {Device-independent security
  of quantum cryptography against collective attacks},}\ }\href {\doibase
  10.1103/PhysRevLett.98.230501} {\bibfield  {journal} {\bibinfo  {journal}
  {Phys. Rev. Lett.}\ }\textbf {\bibinfo {volume} {98}},\ \bibinfo {pages}
  {230501} (\bibinfo {year} {2007})}\BibitemShut {NoStop}%
\bibitem [{\citenamefont {{\v{S}}upi{\'{c}}}\ and\ \citenamefont
  {Bowles}(2020)}]{Supic2020selftestingof}%
  \BibitemOpen
  \bibfield  {author} {\bibinfo {author} {\bibfnamefont {Ivan}\ \bibnamefont
  {{\v{S}}upi{\'{c}}}}\ and\ \bibinfo {author} {\bibfnamefont {Joseph}\
  \bibnamefont {Bowles}},\ }\bibfield  {title} {\enquote {\bibinfo {title}
  {Self-testing of quantum systems: a review},}\ }\href {\doibase
  10.22331/q-2020-09-30-337} {\bibfield  {journal} {\bibinfo  {journal}
  {{Quantum}}\ }\textbf {\bibinfo {volume} {4}},\ \bibinfo {pages} {337}
  (\bibinfo {year} {2020})}\BibitemShut {NoStop}%
\bibitem [{\citenamefont {{\v{S}}upi{\'{c}}}\ \emph {et~al.}(2021)\citenamefont
  {{\v{S}}upi{\'{c}}}, \citenamefont {Cavalcanti},\ and\ \citenamefont
  {Bowles}}]{Supic2021deviceindependent}%
  \BibitemOpen
  \bibfield  {author} {\bibinfo {author} {\bibfnamefont {Ivan}\ \bibnamefont
  {{\v{S}}upi{\'{c}}}}, \bibinfo {author} {\bibfnamefont {Daniel}\ \bibnamefont
  {Cavalcanti}}, \ and\ \bibinfo {author} {\bibfnamefont {Joseph}\ \bibnamefont
  {Bowles}},\ }\bibfield  {title} {\enquote {\bibinfo {title}
  {Device-independent certification of tensor products of quantum states using
  single-copy self-testing protocols},}\ }\href {\doibase
  10.22331/q-2021-03-23-418} {\bibfield  {journal} {\bibinfo  {journal}
  {{Quantum}}\ }\textbf {\bibinfo {volume} {5}},\ \bibinfo {pages} {418}
  (\bibinfo {year} {2021})}\BibitemShut {NoStop}%
\bibitem [{\citenamefont {Bandyopadhyay}\ \emph
  {et~al.}(2014{\natexlab{b}})\citenamefont {Bandyopadhyay}, \citenamefont
  {Jain}, \citenamefont {Oppenheim},\ and\ \citenamefont
  {Perry}}]{Bandyopadyhyay2014Exclusion}%
  \BibitemOpen
  \bibfield  {author} {\bibinfo {author} {\bibfnamefont {Somshubhro}\
  \bibnamefont {Bandyopadhyay}}, \bibinfo {author} {\bibfnamefont {Rahul}\
  \bibnamefont {Jain}}, \bibinfo {author} {\bibfnamefont {Jonathan}\
  \bibnamefont {Oppenheim}}, \ and\ \bibinfo {author} {\bibfnamefont
  {Christopher}\ \bibnamefont {Perry}},\ }\bibfield  {title} {\enquote
  {\bibinfo {title} {Conclusive exclusion of quantum states},}\ }\href
  {\doibase 10.1103/PhysRevA.89.022336} {\bibfield  {journal} {\bibinfo
  {journal} {Phys. Rev. A}\ }\textbf {\bibinfo {volume} {89}},\ \bibinfo
  {pages} {022336} (\bibinfo {year} {2014}{\natexlab{b}})}\BibitemShut
  {NoStop}%
\bibitem [{\citenamefont {Crickmore}(2020)}]{Crickmore2020Elimination}%
  \BibitemOpen
  \bibfield  {author} {\bibinfo {author} {\bibfnamefont {Jonathan}\
  \bibnamefont {Crickmore}},\ }\emph {\bibinfo {title} {Quantum elimination
  measurements}},\ \href@noop {} {Ph.D. thesis},\ \bibinfo  {school}
  {Heriot-Watt University} (\bibinfo {year} {2020})\BibitemShut {NoStop}%
\bibitem [{\citenamefont {Phoenix}\ \emph {et~al.}(2000)\citenamefont
  {Phoenix}, \citenamefont {Barnett},\ and\ \citenamefont
  {Chefles}}]{Phoenix2000}%
  \BibitemOpen
  \bibfield  {author} {\bibinfo {author} {\bibfnamefont {Simon J.~D.}\
  \bibnamefont {Phoenix}}, \bibinfo {author} {\bibfnamefont {Stephen~M.}\
  \bibnamefont {Barnett}}, \ and\ \bibinfo {author} {\bibfnamefont {Anthony}\
  \bibnamefont {Chefles}},\ }\bibfield  {title} {\enquote {\bibinfo {title}
  {Three-state quantum cryptography},}\ }\href {\doibase
  10.1080/09500340008244056} {\bibfield  {journal} {\bibinfo  {journal}
  {Journal of Modern Optics}\ }\textbf {\bibinfo {volume} {47}},\ \bibinfo
  {pages} {507--516} (\bibinfo {year} {2000})}\BibitemShut {NoStop}%
\bibitem [{\citenamefont {Heinosaari}\ \emph {et~al.}(2024)\citenamefont
  {Heinosaari}, \citenamefont {Kerppo}, \citenamefont {Lepp\"aj\"arvi},\ and\
  \citenamefont {Pl\'avala}}]{Heinosaari2024Simple}%
  \BibitemOpen
  \bibfield  {author} {\bibinfo {author} {\bibfnamefont {Teiko}\ \bibnamefont
  {Heinosaari}}, \bibinfo {author} {\bibfnamefont {Oskari}\ \bibnamefont
  {Kerppo}}, \bibinfo {author} {\bibfnamefont {Leevi}\ \bibnamefont
  {Lepp\"aj\"arvi}}, \ and\ \bibinfo {author} {\bibfnamefont {Martin}\
  \bibnamefont {Pl\'avala}},\ }\bibfield  {title} {\enquote {\bibinfo {title}
  {Simple information-processing tasks with unbounded quantum advantage},}\
  }\href {\doibase 10.1103/PhysRevA.109.032627} {\bibfield  {journal} {\bibinfo
   {journal} {Phys. Rev. A}\ }\textbf {\bibinfo {volume} {109}},\ \bibinfo
  {pages} {032627} (\bibinfo {year} {2024})}\BibitemShut {NoStop}%
\bibitem [{\citenamefont {Johnston}\ \emph {et~al.}(2025)\citenamefont
  {Johnston}, \citenamefont {Russo},\ and\ \citenamefont
  {Sikora}}]{Johnston2025tightbounds}%
  \BibitemOpen
  \bibfield  {author} {\bibinfo {author} {\bibfnamefont {Nathaniel}\
  \bibnamefont {Johnston}}, \bibinfo {author} {\bibfnamefont {Vincent}\
  \bibnamefont {Russo}}, \ and\ \bibinfo {author} {\bibfnamefont {Jamie}\
  \bibnamefont {Sikora}},\ }\bibfield  {title} {\enquote {\bibinfo {title}
  {Tight bounds for antidistinguishability and circulant sets of pure quantum
  states},}\ }\href {\doibase 10.22331/q-2025-02-04-1622} {\bibfield  {journal}
  {\bibinfo  {journal} {{Quantum}}\ }\textbf {\bibinfo {volume} {9}},\ \bibinfo
  {pages} {1622} (\bibinfo {year} {2025})}\BibitemShut {NoStop}%
\bibitem [{\citenamefont {Havl\'{\i}\ifmmode~\check{c}\else \v{c}\fi{}ek}\ and\
  \citenamefont {Barrett}(2020)}]{Havlicek2020Exclusion}%
  \BibitemOpen
  \bibfield  {author} {\bibinfo {author} {\bibfnamefont {Vojt\ifmmode
  \check{e}\else~\v{e}\fi{}ch}\ \bibnamefont {Havl\'{\i}\ifmmode~\check{c}\else
  \v{c}\fi{}ek}}\ and\ \bibinfo {author} {\bibfnamefont {Jonathan}\
  \bibnamefont {Barrett}},\ }\bibfield  {title} {\enquote {\bibinfo {title}
  {Simple communication complexity separation from quantum state
  antidistinguishability},}\ }\href {\doibase 10.1103/PhysRevResearch.2.013326}
  {\bibfield  {journal} {\bibinfo  {journal} {Phys. Rev. Res.}\ }\textbf
  {\bibinfo {volume} {2}},\ \bibinfo {pages} {013326} (\bibinfo {year}
  {2020})}\BibitemShut {NoStop}%
\bibitem [{\citenamefont {Srikumar}\ \emph {et~al.}(2024)\citenamefont
  {Srikumar}, \citenamefont {Bartlett},\ and\ \citenamefont
  {Karanjai}}]{Srikumar2024contextuality}%
  \BibitemOpen
  \bibfield  {author} {\bibinfo {author} {\bibfnamefont {Maiyuren}\
  \bibnamefont {Srikumar}}, \bibinfo {author} {\bibfnamefont {Stephen~D.}\
  \bibnamefont {Bartlett}}, \ and\ \bibinfo {author} {\bibfnamefont {Angela}\
  \bibnamefont {Karanjai}},\ }\href {https://arxiv.org/abs/2411.09919}
  {\enquote {\bibinfo {title} {How contextuality and antidistinguishability are
  related},}\ } (\bibinfo {year} {2024}),\ \Eprint
  {http://arxiv.org/abs/2411.09919} {arXiv:2411.09919 [quant-ph]} \BibitemShut
  {NoStop}%
\bibitem [{\citenamefont {Stratton}\ \emph {et~al.}(2024)\citenamefont
  {Stratton}, \citenamefont {Hsieh},\ and\ \citenamefont
  {Skrzypczyk}}]{Stratton2024Choi}%
  \BibitemOpen
  \bibfield  {author} {\bibinfo {author} {\bibfnamefont {Benjamin}\
  \bibnamefont {Stratton}}, \bibinfo {author} {\bibfnamefont {Chung-Yun}\
  \bibnamefont {Hsieh}}, \ and\ \bibinfo {author} {\bibfnamefont {Paul}\
  \bibnamefont {Skrzypczyk}},\ }\bibfield  {title} {\enquote {\bibinfo {title}
  {Operational interpretation of the choi rank through exclusion tasks},}\
  }\href {\doibase 10.1103/PhysRevA.110.L050601} {\bibfield  {journal}
  {\bibinfo  {journal} {Phys. Rev. A}\ }\textbf {\bibinfo {volume} {110}},\
  \bibinfo {pages} {L050601} (\bibinfo {year} {2024})}\BibitemShut {NoStop}%
\bibitem [{\citenamefont {Duan}\ \emph {et~al.}(2016)\citenamefont {Duan},
  \citenamefont {Severini},\ and\ \citenamefont {Winter}}]{Duan2016}%
  \BibitemOpen
  \bibfield  {author} {\bibinfo {author} {\bibfnamefont {Runyao}\ \bibnamefont
  {Duan}}, \bibinfo {author} {\bibfnamefont {Simone}\ \bibnamefont {Severini}},
  \ and\ \bibinfo {author} {\bibfnamefont {Andreas}\ \bibnamefont {Winter}},\
  }\bibfield  {title} {\enquote {\bibinfo {title} {On zero-error communication
  via quantum channels in the presence of noiseless feedback},}\ }\href@noop {}
  {\bibfield  {journal} {\bibinfo  {journal} {IEEE Trans. Inf. Theory}\
  }\textbf {\bibinfo {volume} {62}},\ \bibinfo {pages} {5260--5277} (\bibinfo
  {year} {2016})}\BibitemShut {NoStop}%
\bibitem [{\citenamefont {Chiribella}\ \emph {et~al.}(2025)\citenamefont
  {Chiribella}, \citenamefont {Roy}, \citenamefont {Guha},\ and\ \citenamefont
  {Saha}}]{Chiribella2025}%
  \BibitemOpen
  \bibfield  {author} {\bibinfo {author} {\bibfnamefont {Giulio}\ \bibnamefont
  {Chiribella}}, \bibinfo {author} {\bibfnamefont {Saptarshi}\ \bibnamefont
  {Roy}}, \bibinfo {author} {\bibfnamefont {Tamal}\ \bibnamefont {Guha}}, \
  and\ \bibinfo {author} {\bibfnamefont {Sutapa}\ \bibnamefont {Saha}},\
  }\bibfield  {title} {\enquote {\bibinfo {title} {Communication power of a
  noisy qubit},}\ }\href {\doibase 10.1103/PhysRevLett.134.080803} {\bibfield
  {journal} {\bibinfo  {journal} {Phys. Rev. Lett.}\ }\textbf {\bibinfo
  {volume} {134}},\ \bibinfo {pages} {080803} (\bibinfo {year}
  {2025})}\BibitemShut {NoStop}%
\bibitem [{\citenamefont {Saha}\ \emph {et~al.}(2023)\citenamefont {Saha},
  \citenamefont {Das}, \citenamefont {Das}, \citenamefont {Bhattacharya},\ and\
  \citenamefont {Majumdar}}]{deba+PRA2023}%
  \BibitemOpen
  \bibfield  {author} {\bibinfo {author} {\bibfnamefont {Debashis}\
  \bibnamefont {Saha}}, \bibinfo {author} {\bibfnamefont {Debarshi}\
  \bibnamefont {Das}}, \bibinfo {author} {\bibfnamefont {Arun~Kumar}\
  \bibnamefont {Das}}, \bibinfo {author} {\bibfnamefont {Bihalan}\ \bibnamefont
  {Bhattacharya}}, \ and\ \bibinfo {author} {\bibfnamefont {A.~S.}\
  \bibnamefont {Majumdar}},\ }\bibfield  {title} {\enquote {\bibinfo {title}
  {Measurement incompatibility and quantum advantage in communication},}\
  }\href {\doibase 10.1103/PhysRevA.107.062210} {\bibfield  {journal} {\bibinfo
   {journal} {Phys. Rev. A}\ }\textbf {\bibinfo {volume} {107}},\ \bibinfo
  {pages} {062210} (\bibinfo {year} {2023})}\BibitemShut {NoStop}%
\bibitem [{\citenamefont {Kerppo}(2023)}]{Kerppo2023Communication}%
  \BibitemOpen
  \bibfield  {author} {\bibinfo {author} {\bibfnamefont {Oskari}\ \bibnamefont
  {Kerppo}},\ }\emph {\bibinfo {title} {Quantum communication tasks}},\
  \href@noop {} {Ph.D. thesis},\ \bibinfo  {school} {University of Turku}
  (\bibinfo {year} {2023})\BibitemShut {NoStop}%
\bibitem [{\citenamefont {Frenkel}\ and\ \citenamefont
  {Weiner}(2015)}]{Frenkel2015Classical}%
  \BibitemOpen
  \bibfield  {author} {\bibinfo {author} {\bibfnamefont {P.~E.}\ \bibnamefont
  {Frenkel}}\ and\ \bibinfo {author} {\bibfnamefont {M.}~\bibnamefont
  {Weiner}},\ }\bibfield  {title} {\enquote {\bibinfo {title} {Classical
  information storage in an n-level quantum system},}\ }\href {\doibase
  10.1007/s00220-015-2463-0} {\bibfield  {journal} {\bibinfo  {journal}
  {Commun. Math. Phys.}\ }\textbf {\bibinfo {volume} {340}},\ \bibinfo {pages}
  {563} (\bibinfo {year} {2015})}\BibitemShut {NoStop}%
\bibitem [{\citenamefont {Cohen}\ and\ \citenamefont
  {Rothblum}(1993)}]{Cohen1993Nonnegative}%
  \BibitemOpen
  \bibfield  {author} {\bibinfo {author} {\bibfnamefont {J.~E.}\ \bibnamefont
  {Cohen}}\ and\ \bibinfo {author} {\bibfnamefont {U.~G.}\ \bibnamefont
  {Rothblum}},\ }\bibfield  {title} {\enquote {\bibinfo {title} {Nonnegative
  ranks, decompositions, and factorizations of nonnegative matrices},}\ }\href
  {\doibase 10.1016/0024-3795(93)90224-C} {\bibfield  {journal} {\bibinfo
  {journal} {Linear Algebra Appl.}\ }\textbf {\bibinfo {volume} {190}},\
  \bibinfo {pages} {149} (\bibinfo {year} {1993})}\BibitemShut {NoStop}%
\bibitem [{\citenamefont {Fawzi}\ \emph {et~al.}(2015)\citenamefont {Fawzi},
  \citenamefont {Gouveia}, \citenamefont {Parrilo}, \citenamefont {Robinson},\
  and\ \citenamefont {Thomas}}]{Fawzi2015Positive}%
  \BibitemOpen
  \bibfield  {author} {\bibinfo {author} {\bibfnamefont {H.}~\bibnamefont
  {Fawzi}}, \bibinfo {author} {\bibfnamefont {J.}~\bibnamefont {Gouveia}},
  \bibinfo {author} {\bibfnamefont {P.~A.}\ \bibnamefont {Parrilo}}, \bibinfo
  {author} {\bibfnamefont {R.~Z.}\ \bibnamefont {Robinson}}, \ and\ \bibinfo
  {author} {\bibfnamefont {R.~R.}\ \bibnamefont {Thomas}},\ }\bibfield  {title}
  {\enquote {\bibinfo {title} {Positive semidefinite rank},}\ }\href {\doibase
  10.1007/s10107-015-0922-1} {\bibfield  {journal} {\bibinfo  {journal} {Math.
  Program.}\ }\textbf {\bibinfo {volume} {153}},\ \bibinfo {pages} {133}
  (\bibinfo {year} {2015})}\BibitemShut {NoStop}%
\bibitem [{\citenamefont {Lee}\ \emph {et~al.}(2017)\citenamefont {Lee},
  \citenamefont {Wei},\ and\ \citenamefont {de~Wolf}}]{Lee2017Some}%
  \BibitemOpen
  \bibfield  {author} {\bibinfo {author} {\bibfnamefont {T.}~\bibnamefont
  {Lee}}, \bibinfo {author} {\bibfnamefont {Z.}~\bibnamefont {Wei}}, \ and\
  \bibinfo {author} {\bibfnamefont {R.}~\bibnamefont {de~Wolf}},\ }\bibfield
  {title} {\enquote {\bibinfo {title} {Some upper and lower bounds on
  psd-rank},}\ }\href {\doibase 10.1007/s10107-016-1052-0} {\bibfield
  {journal} {\bibinfo  {journal} {Math. Program.}\ }\textbf {\bibinfo {volume}
  {162}},\ \bibinfo {pages} {495} (\bibinfo {year} {2017})}\BibitemShut
  {NoStop}%
\end{thebibliography}%

\end{document}